%% file: rsa-standards.tex
\newif\ifllncs
\llncstrue
\llncsfalse
\ifllncs
\documentclass{llncs}
\usepackage{cc4llncs}
\makeatletter
\@longarticletrue
\AtBeginDocument{\@adjustnumbering}
\makeatother
\usepackage[T1]{fontenc}
\usepackage[latin1]{inputenc}
\usepackage[german,polutonikogreek,english]{babel}
\else
\documentclass[11pt,noccpublish,textarea=cc1,
multilingual,german,polutonikogreek,english]{cc}
\fi

\usepackage[final,fullversion]{versions}

\newif\ifwatermark%
\watermarktrue%
\watermarkfalse%

\newif\ifanonymous%
\anonymoustrue%
\anonymousfalse%

\newif\ifarxivid
\arxividtrue
\arxividfalse
\ifarxivid

\else

\fi

\ifdraft
\usepackage[draft,showlabels]{develop}
\def\rightflag#1{}

\else
\fi

\ifwatermark
\usepackage{watermark}
\definecolor{watermarkcolor}{cmyk}{.08,.052,0,0}
\watermark{\rput[tl]{60}(0,-\textheight){%
      \psscaleboxto(22cm,15mm){\Huge\bf\color{watermarkcolor}
    DRAFT
    ---
    Do not distribute!
  }}}

\fi

\ifdraft
\else
\fi

\usepackage{times}
\usepackage{wasysym}
\usepackage{tilde}
\usepackage{maps}
\usepackage{mathell}
\usepackage{set}
\usepackage{multicol}
\usepackage{graphicx}
\usepackage{sumint}
\usepackage{pifont}
\usepackage{weakhyph}
\usepackage{rounded}
\usepackage{multirow}
\usepackage{bigOetc}

\usepackage{rsa-notions}
\everydrawnotion{%
}

\makeatletter\def\thefootnote{\@fnsymbol\c@footnote}\makeatother
\usepackage{endnotes}

\ifdraft
\AtEndDocument{%
  \clearpage\appendix
  \tableofcontents
  \section{TODO}
  
  \begin{itemize}
  \item ...
  \end{itemize}
  \begingroup
  \def\enoteformat{\rightskip0pt \leftskip0pt \parindent=1.8em
    \def\xxx{1}\expandafter\ifx\csname @theenmark\endcsname\xxx\else
    \addvspace{4ex}%
    \fi
    \leavevmode\llap{\makeenmark}}
  \theendnotes
  \endgroup
}
\fi

\def\swap#1#2{\let\swaptmp#1\let#1#2\let#2\swaptmp}
\swap\epsilon\varepsilon
\swap\theta\vartheta
\swap\rho\varrho

\newcommand{\N}{\mathbb{N}}
\newcommand{\Z}{\mathbb{Z}}
\newcommand{\R}{\mathbb{R}}
\newcommand{\PR}{\mathbb{P}}

\newcommand{\C}{\ensuremath{\mathcal{C}}}

\newcommand{\floor}[1]{\left\lfloor#1\right\rfloor}
\newcommand{\ceil}[1]{\left\lceil#1\right\rceil}
\newcommand{\quo}[2]{\floor{\frac{#1}{#2}}}

\renewcommand{\phi}{\varphi}

\makeatletter
\def\treatparameter{%
  \@ifnextchar({\treatparameterx}{%
    \@ifnextchar<{\treatparameterxi}{}}}
\makeatother
\def\treatparameterx(#1){\left\lgroup#1\right\rgroup}
\def\treatparameterx(#1){\left(#1\right)}
\def\treatparameterxi<#1>{\left\ldbrack#1\right\rdbrack}
\def\treatparameterxi<#1>{\left\langle#1\right\rangle}

\newcommand{\Qclean}[2][B,C]{\pi_{#1}^{#2}\treatparameter}
\newcommand{\Q}[2][B,C]{\kappa_{#1}^{#2}\treatparameter}

\newcommand{\QDM}[1][r]{\C_{#1}\treatparameter}
\newcommand{\RSASET}{\mathcal{A}\treatparameter}
\newcommand{\RSADMSET}[1][r]{\mathcal{A}^{~DM(#1)}}
\newcommand{\RSAFIXSET}[1][r,\sigma]{\mathcal{A}^{~FB(#1)}}
\iffullversion
\newcommand{\RSAALGSET}[1][r]{\mathcal{A}^{~ALG_{1}(#1)}}
\newcommand{\RSAALGVARSET}[1][r,\sigma]{\mathcal{A}^{~ALG(#1)}}
\newcommand{\RSAALG}[1][r]{\#\RSAALGSET[#1]\treatparameter}
\newcommand{\RSAALGVAR}[1][r,\sigma]{\#\RSAALGVARSET[#1]\treatparameter}
\else
\newcommand{\RSAALGVARSET}[1][r,\sigma]{\mathcal{A}^{~ALG(#1)}}
\newcommand{\RSAALGVAR}[1][r]{\#\RSAALGVARSET[#1]\treatparameter}
\fi

\newcommand{\RSAMAXSET}[1][r, c_1]{\mathcal{M}^{#1}}
\newcommand{\RSA}{\#\RSASET\treatparameter}
\newcommand{\RSADM}[1][r]{\#\RSADMSET[#1]\treatparameter}
\newcommand{\RSAFIX}[1][r,\sigma]{\#\RSAFIXSET[#1]\treatparameter}
\newcommand{\RSAMAX}[1][r, c_1]{\#\RSAMAXSET[#1]\treatparameter}

\newpsstyle{notionarea1}{%
  linewidth=.4pt,%
  fillcolor=cyan,%
  fillcolor=lightgray,%
  fillstyle=solid%
}
\newpsstyle{notionarea2}{%
  linewidth=.4pt,%
  fillstyle=vlines,%
  hatchcolor=red,%
  hatchangle=45,%
  hatchsep=2pt,%
  hatchwidth=.4pt}%
\newpsstyle{bannerarea}{%
  linewidth=.4pt,%
  linestyle=dashed,%
  fillstyle=vlines,%
  hatchcolor=red,%
  hatchangle=-45,%
  hatchsep=2pt,%
  hatchwidth=.4pt}%
\newpsstyle{notionarea}{style=notionarea1}
\newpsstyle{notionsplit}{linestyle=solid}

\newcommand{\RSAGENERALdraw}{%
  \psccurve[style=notionarea,showpoints=false]%
  (.1,.1)(1.1,-.9)(1.9,-.6)(1.5,.3)(.9,.9)(.4,1.2)(-.5,.9)(-.1,.5)
}

\newcommand{\RSADMdraw}{%
  \pspolygon[style=notionarea](0,2)(-1,1)(1,-1)(2,0)%
  \psline[style=notionsplit](0,0)(1,1)(1,-1)%
}
\newcommand{\RSAFIXdraw}[1]{%
  \pspolygon[style=notionarea](#1,-#1)(1,-#1)(1,1)(-#1,1)(-#1,#1)%
  \psline[style=notionsplit](#1,1)(#1,-#1)%
}
\newcommand{\RSAALGdraw}{%
  \pspolygon[style=notionarea](2,0)(0,2)(0,0)(2,-2)%
  \psline[style=notionsplit](2,0)(0,0)%
}
\newcommand{\RSAMAXdraw}{%
  \pspolygon[style=notionarea](2,0)(0,2)(-2,2)(2,-2)%
  \psline[style=notionsplit](2,0)(0,0)%
  \psline[style=notionsplit](0,0)(0,2)%
}
\newcommand{\RSAALGVARdraw}[1]{%
  \rput(-#1,#1){\RSAALGdraw}%
}

\def\primesum#1#2#3#4{\sum_{#2 \in \P{#3,#4}} #1}%
\def\friint#1#2#3#4#5{\iint_{#4}^{#5} #1 \,~d\! #3 \,~d\! #2}%

\renewcommand{\P}[1]{{\mathbb{P}\mathbin{\cap}\left]#1\right]}}

\def\Frac#1#2{{\mathpalette{\let\cursize}{%
      \raise0.5ex\hbox{$\cursize#1$}/\lower0.5ex\hbox{$\cursize#2$}}}}%

\def\Land{\;\land\;}

\newcommand{\intrangeleft}{B}
\newcommand{\intrangeright}{C}

\makeatletter
\def\casesincreaseheight#1{%
  \expandafter\global\expandafter\setbox\csname
  @arstrutbox\endcsname\hbox{%
    \vrule height#1\ht\strutbox depth#1\dp\strutbox width0pt}}

\makeatother

\iffullversion
\title{Notions for RSA integers}%
\subtitle{
    Full version
}
\else
\title{Analyzing standards for RSA integers}
\fi

\ifanonymous
\author{
  Anonymous.
}
\else\ifllncs
\author{Daniel Loebenberger  \and Michael N{\"{u}}sken
}
\institute{
  B-IT, University of Bonn\\
  \texttt{\{daniel,nuesken\}@bit.uni-bonn.de}
}
\index{Loebenberger, Daniel}
\index{N{\"{u}}sken, Michael}
\else
\author{Daniel Loebenberger\\
  B-IT cosec, University of Bonn\\
  Dahlmannstr. 2\\
  53113 Bonn\\
  \email{daniel@bit.uni-bonn.de}
  \and
  Michael N{\"{u}}sken\\
  B-IT cosec, University of Bonn\\
  Dahlmannstr. 2\\
  53113 Bonn\\
  \email{nuesken@bit.uni-bonn.de}
}
\fi\fi

\ifllncs
\def\x{
\begin{document}\maketitle}\def\y{}
\else
\def\x{}\def\y{\begin{document}}
\fi
\x
\begin{abstract}
  The key-generation algorithm for the RSA cryptosystem is specified
  in several standards, such as PKCS\#1, IEEE~1363-2000, FIPS~186-3,
  ANSI~X9.44, or ISO/IEC~18033-2.  All of them substantially differ in
  their requirements. This indicates that for computing a ``secure''
  RSA modulus it does not matter how exactly one generates RSA
  integers.  In this work we show that this is indeed the case to a
  large extend: First, we give a theoretical framework that will
  enable us to easily compute the entropy of the output distribution
  of the considered standards and show that it is comparatively
  high. To do so, we compute for each standard the number of integers
  they define (up to an error of very small order) and discuss
  different methods of generating integers of a specific form. Second,
  we show that factoring such integers is hard, provided factoring a
  product of two primes of similar size is hard.

  \ifllncs
  Keywords: RSA integer, output entropy, reduction. ANSI~X9.44,
  FIPS~186-3, IEEE~1363-2000, ISO/IEC~18033-2, NESSIE,
  PKCS\#1\iffullversion, GnuPG, OpenSSL, Open\-Swan,
    SSH.\else .\fi
  \fi
\end{abstract}
\ifllncs\else
\begin{keywords}
  RSA integer, output entropy, reduction. ANSI~X9.44,
  FIPS~186-3, IEEE~1363-2000, ISO/IEC~18033-2, NESSIE,
  PKCS\#1\begin{fullversion}[.], GnuPG, OpenSSL, Open\-Swan,
    SSH.\end{fullversion}
\end{keywords}
\begin{subject}
  MSC2010: 94A60,   
  11A51, 
  11N25, 
  94A17 
\end{subject}
\fi
\y


\section{Introduction}
\noindent
An \emph{RSA integer} is an integer that is suitable as a modulus for
the RSA cryptosystem as proposed by \cite*{rivsha77,rivsha78}:
\begin{quotation}
  \noindent ``You first compute $n$ as the product of two primes $p$
  and $q$:
  $$
  n = p \cdot q.
  $$
  These primes are very large, 'random' primes. Although you will make
  $n$ public, the factors $p$ and $q$ will be effectively hidden from
  everyone else due to the enormous difficulty of factoring $n$.''
\end{quotation}
Also in earlier literature such as \cite{ell70} or \cite{coc73} one
does not find any further restrictions. In subsequent literature
people define RSA integers similarly to
\citeauthor*{rivsha77}\begin{fullversion}[, while sometimes additional
  safety tests are performed. ]: \cite*{crapom01} note that it is
  ``fashionable to select approximately equal primes but sometimes one
  runs some further safety tests''. In more applied works such as
  \cite{sch96h} or \cite{menoor97} one can read that for maximum
  security one chooses two (distinct) primes of equal length. Also
  \cite*{gatger03} follow a similar approach.
\end{fullversion}
\begin{fullversion}
  On suggestion of B. de~Weger, \cite{decmor08} define an RSA integer
  to be a product of two primes $p$ and $q$ such that $p < q < r p$
  for some parameter $r \in \R_{>1}$.
\end{fullversion}
Real world 
implementations, however, require \emph{concrete
  algorithms} that specify in detail how to generate RSA
integers. This has led to a variety of standards, notably the
standards PKCS\#1 \citep{jonkal03}, ISO~18033-2 \citep{iso180332},
IEEE~1363-2000 \citep{araarn00}, ANSI~X9.44 \citep{ansi944},
FIPS~186-3 \citep{fips1863}, the standard of the RSA foundation
\citep{rsa00}, the standard set by the German Bundesnetzagentur
\citep{woh08}, and the standard resulting from the European NESSIE
project \citep{prebir03}.  All of those standards define more or less
precisely how to generate RSA integers and all of them have
substantially different requirements. This
reflects the intuition that it does not really matter how one selects
the prime factors in detail, the resulting RSA modulus will do its
job.  But what is needed to show that this is really the case?

Following \cite{bradam93} a quality measure of a generator is the
entropy of its output distribution.  In abuse of language we will most
of the time talk about the \emph{output entropy} of an algorithm.  To
compute it, we need estimates of the probability that a certain
outcome is produced. This in turn needs a thorough analysis of how one
generates RSA integers of a specific form. If we can show that the
outcome of the algorithm is roughly uniformly distributed, the output
entropy is closely related to the count of RSA integers it can
produce. It will turn out that in all reasonable setups this count is
essentially determined by the desired length of the
output\begin{fullversion}[. ], see
  \ref{sec:CompareNotions}. \end{fullversion} For primality tests
there are several results in this direction (see for example
\citealt{joypai06}) but we are not aware of any related work analyzing
the output entropy of algorithms for generating RSA integers.


Another requirement for the algorithm is that the output should be
`hard to factor'.  Since this statement does not even make sense for a
single integer, this means that one has to show that the restrictions
on the shape of the integers the algorithm produces do not introduce
any further possibilities for an attacker. To prove this, a
\emph{reduction} has to be given that reduces the problem of factoring
the output to the problem of factoring a product of two primes of
similar size, see \ref{sec:algoConsid}. Also there it is necessary to
have results on the count of RSA integers of a specific form to make
the reduction work.  As for the entropy estimations, we do not know
any related work on this.
\ifanonymous\else A conference version of this article, focusing on
the analysis of standardized RSA key-generators only, was published in
\ifanonymous Anonymous (2011)\else\cite{loenus11ae}\fi.\fi

In the following section we will develop a formal framework that can
handle all possible definitions for RSA integers. After discussing the
necessary number theoretic tools in \ref{sec:toolbox}, we give
explicit formul\ae{} for the count of such integers which will be used
later for entropy estimations of the various standards for RSA
integers. In \ref{sec:NotionsInParticular} we show how our general
framework can be instantiated, yielding natural definitions for
several types of RSA integers (as used later in the standards).
\begin{fullversion}%
  The section afterwards compares in more detail the relations of the
  different notions.
\end{fullversion}%
\ref{sec:genProperly} gives a short overview on generic constructions
for fast algorithms that generate such integers almost uniformly. At
this point we will have described all necessary techniques to compute
the output entropy, which we discuss in \ref{sec:entropy}. The
following section resolves the second question described above by
giving a reduction from factoring special types of RSA integers to
factoring a product of two primes of similar size. We finish by
applying our results to various standards for RSA integers in
\ref{sec:concImplem}.


\begin{shortversion}
  We omitted here most of the number theoretic details. \ifanonymous
  Proofs of those theorems can be found in the extended version of
  this paper. \else For the proofs of those theorems see \ifanonymous
  Anonymous (2011)\else\cite{loenus11a}\fi. \fi Note that for ease of
  comparison, we have retained the numbering of the extended version.
\end{shortversion}

\section{RSA integers in general}

If one generates an RSA integer it is necessary to select for each
choice of the security parameter the prime factors from a certain
region. This security parameter is typically an integer $k$ that
specifies (roughly) the size of the output. We use a more general
definition by asking for integers from the interval $]x/r, x]$, given
a \emph{real} bound $x$ and a parameter $r$ (possibly depending on
$x$). Clearly, this can also be used to model the former selection
process by setting $x = 2^{k}-1$ and $r = 2$. Let us in general
introduce a \emph{notion of RSA integers with tolerance~$r$} as a
family%
\begin{shortversion}
  \hanghere[\baselineskip]{\psset{unit=1.9mm}\drawnotion{8}{3}{%
      \RSAGENERALdraw \rput{315}(0.9,0.1){\tiny $(~ln \RSASET)_x$}%
    }}
\end{shortversion}
$$
\RSASET := \left\langle \RSASET_{x} \right\rangle_{x \in \R_{>1}}
$$
of subsets of the positive quadrant $\R_{>1}^{2}$, where for every
$x\in\R_{>1}$
$$
\RSASET_{x} \subseteq \Set{(y,z) \in \R_{>1}^{2}; \frac{x}{r} < y z
  \leq x}.
$$
The tolerance~$r$ shall always be larger than~$1$.  We allow here that
$r$ varies with~$x$, which of course includes the case when
$r$ is a constant.  Typical values used for RSA are $r=2$ or $r=4$
which fix the bit-length of the modulus more or less. 
\begin{shortversion}
  We can ---~for a fixed choice of parameters~--- easily visualize any
  notion of RSA integers by the corresponding region $\RSASET_{x}$ in
  the $(y,z)$-plane.  It is favorable to look at these regions in
  logarithmic scale. We write $y = e^{\upsilon}$ and $z = e^{\zeta}$
  and denote by $(~ln\RSASET)_{x}$ the region in the
  $(\upsilon,\zeta)$-plane corresponding to the region~$\RSASET_{x}$
  in the $(y,z)$-plane, formally\ $(\upsilon,\zeta) \in (~ln
  \RSASET)_{x} :\Leftrightarrow (y,z) \in \RSASET_{x}$%
  \begin{fullversion}[. ]%
    , we obtain a picture like in \ref{fig:RSAnotionsGeneral}.
  \end{fullversion}%
\end{shortversion}
Now an \emph{$\RSASET$\=integer $n$ of size~$x$} ---~for use as a
modulus in RSA~--- is a product $n = p q$ of a prime pair $(p,q) \in
\RSASET_{x} \cap (\PR\times\PR)$, where $\PR$ denotes the set of
primes.
\begin{fullversion}%
  \begin{figure}[t]
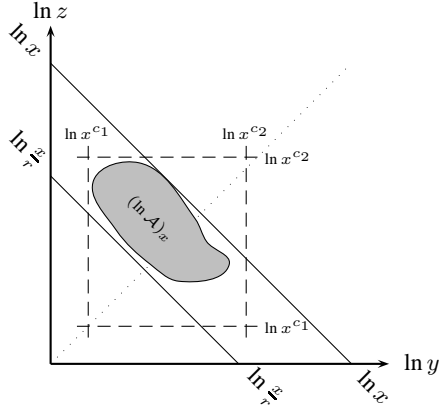

    \centering
    { \psset{unit=5mm}%
      \drawnotion[{%
        \DNaddlnlabels
      }]{8}{3}{%
        \RSAGENERALdraw \psline[linewidth=.4pt,
        linestyle=dashed](-.8,-1.2)(-.8,2.2) \psline[linewidth=.4pt,
        linestyle=dashed](-1,-1)(2.2,-1) \psline[linewidth=.4pt,
        linestyle=dashed](2,-1.2)(2,2.2) \psline[linewidth=.4pt,
        linestyle=dashed](-1,2)(2.2,2) \psset{labelsep=1mm}
        \uput[90]{0}(2,-1.2){\tiny $~ln x^{c_1}$}
        \uput[90]{0}(-.8,-1.2){\tiny $~ln x^{c_2}$}
        \uput[0]{0}(-1,-1){\tiny $~ln x^{c_2}$}
        \uput[0]{0}(-1,2){\tiny $~ln x^{c_1}$}
        \rput{315}(0.9,0.1){\tiny $(~ln \RSASET)_x$}%
      }}%
    \caption{%
      A generic notion of RSA integers with tolerance~$r$.  The gray
      area shows the parts of the $(~ln y,~ln z)$-plane which is
      counted.  It lies between the tolerance bounds $~ln x$ and $~ln
      \frac{x}{r}$.  The dashed lines show boundaries as imposed by
      $[c_{1},c_{2}]$-balanced.  The dotted diagonal marks the
      criterion for symmetry.}
    \label{fig:RSAnotionsGeneral}
  \end{figure}%
\end{fullversion}%
They are counted by the associated \emph{prime~pair counting
  function~$\RSA$} for the notion $\RSASET$:
$$
\map[\RSA]{\R_{>1}}{\N}{x}{\#\Set{(p,q) \in \PR \times \PR; (p,q) \in
    \RSASET_{x}}.}
$$
Thus every $\RSASET$\=integer $n = p q$ is counted once or twice in
$\RSA(x)$ depending on whether only $(p,q) \in \RSASET_{x}$ or also
$(q,p) \in \RSASET_{x}$, respectively.  We call a notion
\emph{symmetric} if for all choices of the parameters the
corresponding area in the $(y,z)$-plane is symmetric with respect to
the main diagonal, i.e.\ that $(y,z) \in \RSASET_{x}$ implies also
$(z,y) \in \RSASET_{x}$.  If to the contrary $(y,z) \in \RSASET_{x}$
implies $(z,y) \notin \RSASET_{x}$ we call the notion
\emph{antisymmetric}.  When we are only interested in the associated
RSA integers we can always require symmetry or antisymmetry, yet many
algorithms proceed in an asymmetric way.

Note that varying $r$ do not occur in standards and implementations
for RSA integers, analyzed in \ref{sec:concImplem}.  However, there
are still quite natural notions in which a varying $r$ occurs:
Consider for example the notion where the primes $p$, $q$ are selected
from the interval $[x^{1/4}, x^{1/2}]$.  Then we obtain the product $p
q \in [x^{1/2}, x]$.  This corresponds to the notion discussed in
\ref{sec:fixedBound} with $r = \sqrt{x}$.  Indeed, all the counting
theorems in \ref{sec:NotionsInParticular} can handle such large $r$.
However, the error term is correspondingly large.

Certainly, we will also need restrictions on the shape of the area we
are analyzing: If one considers any notion of RSA integers and throws
out exactly the prime pairs one would be left with a prime-pair-free
region and any approximation for the count of such a notion based on
the area would necessarily have a tremendously large error
term. However, for practical applications it turns out that it is
enough to consider regions of a very specific form.  Actually, we will
most of the time have regions whose boundary can be described by
graphs of certain smooth functions\iffullversion, see
\ref{def:notion-monotone}\fi.
\begin{shortversion}
  In the following, we call notions having such boundaries
  \emph{monotone}. A more detailed explanation of the restrictions we
  have to impose to make the number-theoretic work sound can be found
  in the extended version\ifanonymous\else\ \cite{loenus11a}\fi.
\end{shortversion}

For RSA, people usually prefer two prime factors of roughly the same
size, where size is understood as bit length.  Accordingly, we call a
notion of RSA integers \emph{$[c_{1},c_{2}]$-balanced} iff
additionally for every $x\in\R_{>1}$
$$
\RSASET_{x} \subseteq \Set{(y,z)\in\R_{>1}^{2}; y, z \in \left[
    x^{c_{1}}, x^{c_{2}} \right]},
$$
where $0 < c_{1} \leq c_{2}$ can be thought of as constants or ---
more generally --- as smooth functions in $x$ defining the amount of
allowed divergence subject to the side condition that $x^{c_{1}}$
tends to infinity when $x$ grows.  If $c_{1} > \frac{1}{2}$ then
$\mathcal{A}_{x}$ is empty, so we will usually assume $c_{1} \leq
\frac{1}{2}$.  In order to prevent trial division from being a
successful attacker it would be sufficient to require $y, z \in
\bigOmega{~ln^{k} x}$ for every $k \in \N$.
Our stronger requirement still seems reasonable and indeed equals the
condition \cite{mau95} required for secure RSA moduli, as the
supposedly most difficult factoring challenges stay within the range
of our attention.  As a side-effect this greatly simplifies our
approximations later. The German Bundesnetzagentur uses a very similar
restriction in their algorithm catalog \citep{woh08}.
\begin{fullversion}
  We can ---~for a fixed choice of parameters~--- easily visualize any
  notion of RSA integers by the corresponding region $\RSASET_{x}$ in
  the $(y,z)$-plane.  It is favorable to look at these regions in
  logarithmic scale: writing $y = e^{\upsilon}$ and $z = e^{\zeta}$,
  we depict the region $(~ln\RSASET)_{x}$ in the
  $(\upsilon,\zeta)$-plane corresponding to the region~$\RSASET_{x}$
  in the $(y,z)$-plane, i.e.\ $(\upsilon,\zeta) \in (~ln \RSASET)_{x}
  :\Leftrightarrow (y,z) \in \RSASET_{x}$.%
  \begin{fullversion}%
    We obtain a picture like in \ref{fig:RSAnotionsGeneral}.
  \end{fullversion}%
\end{fullversion}

Often the considered integers $n = pq$ are also subject to further
side conditions, like $~gcd((p-1)(q-1), e) = 1$ for some fixed public
RSA exponent $e$. Most of the number theoretic work below can easily
be adapted, but for simplicity of exposition we will often present our
results without those further restrictions and just point out when
necessary how to incorporate such additional properties.

In \cite{woh08} it is additionally required that the primes $p$ and
$q$ are not too close to each other.  We ignore this issue here, since
the probability that two primes are \emph{very close} to each other
would be tiny if the notion from which $(p,q)$ was selected is
sufficiently large. If necessary, we are able to modify our notions
such that also this requirement is met.

\begin{fullversion}
  \begin{figure}[t]
    \centering
    \input{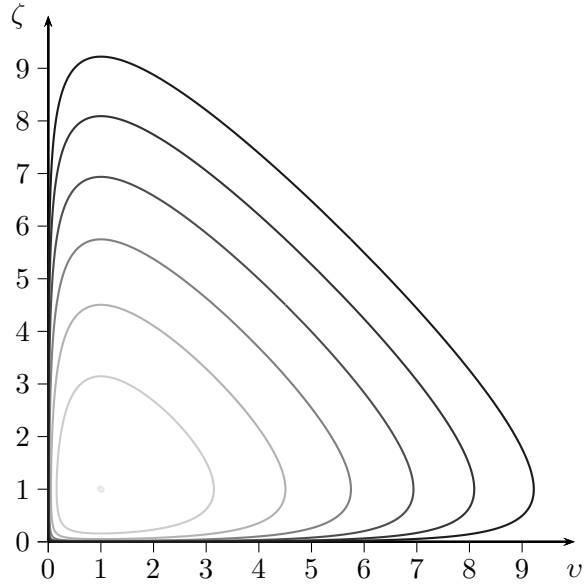}
    \caption{\label{fig:density} Levels $~e^{k}$ of the function
      $\frac{~e^{\upsilon+\zeta}}{\upsilon \zeta}$ for $k \in
      \{2+\epsilon, 3, \ldots, 8\}$. The darker the line the higher is
      the value of $k$.}
  \end{figure}

  In order to count the number of $\RSASET$\=integers we have to
  evaluate
  $$
  \RSA(x) = \sum_{\substack{(p,q) \in \RSASET_{x} \\ p,q \in \PR}} 1.
  $$
  If we follow the intuitive view that a randomly generated number~$n$
  is prime with probability $\frac{1}{~ln n}$, we expect that we have
  to evaluate integrals like
  $$
  \friint{\frac{1}{~ln y ~ln z}}{y}{z}{\RSASET_{x}}{},
  $$
  while carefully considering the error between those integrals and
  the above sums.
  In logarithmic scale we obtain expressions of the form $
  \friint{\frac{e^{\upsilon+\zeta}}{\upsilon
      \zeta}}{\upsilon}{\zeta}{(~ln \RSASET)_{x}}{}.  $ To get an
  understanding of these functions, in \ref{fig:density} some contour
  lines of the inner function are depicted.  From the figure we
  observe that pairs $(\upsilon, \zeta)$ where $\upsilon+\zeta$ is
  large have a higher weight in the overall count.
\end{fullversion}%

As we usually deal with balanced notions the considered regions are
somewhat centered around the main diagonal. We will show in
\ref{sec:algoConsid} that if factoring products of two primes is hard
then it is also hard to factor integers generated from such notions.

\section{Toolbox}\label{sec:toolbox}

We will now develop the necessary number theoretic concepts to obtain
formul\ae{} for the count of RSA integers that will later help us to
estimate the output entropy of the various standards for RSA
integers. In related articles, like \cite{decmor08} one finds counts
for \emph{one particular} definition of RSA integers.  We believe that
in the work presented here for the first time a sufficiently general
theorem is established that allows to compute the number of RSA
integers for \emph{all} reasonable definitions.

We assume the Riemann hypothesis throughout the entire paper.  The
main terms are the same without this assumption, but the error bounds
one obtains are then much weaker.
\begin{fullversion}[\refstepcounter{equation}\refstepcounter{equation}]%
  We use the following version of the prime number theorem:
  \begin{namedtheorem}{Prime number theorem}[\Citealt{koc01},
    \cite{sch76a}]
    \label{pnt-schoenfeld}
    \leavevmode\hfil
    If (and only if) the Riemann hypothesis holds, then for
    $x\geq2657$
    $$
    \left| \pi(x) - ~li(x) \right| < \frac{1}{8\pi} \sqrt{x} ~ln x,
    $$
    where $~li(x) := \preint{t}{0}{x} \frac{\postint}{~ln t}$.
  \end{namedtheorem}
  We first state a quite technical lemma that enables us to do our
  approximations:%
  \begin{lemma}[Prime sum approximation]\label{res:recapprox}
    Let $f$, $\widetilde{f}$, $\widehat{f}$ be functions $\shortmap{
      [\intrangeleft,\intrangeright] }{\R_{>1}}$, where $B,C \in
    \R_{>1}$ such that $\widetilde{f}$ and $\widehat{f}$ are piecewise
    continuous, $\widetilde{f}+\widehat{f}$ is either weakly
    decreasing, weakly increasing, or constant, and for $p \in [B,C]$
    we have the estimate
    \begin{equation*}
      \left| f(p) - \widetilde{f}(p) \right| \leq \widehat{f}(p).
    \end{equation*}
    Further, let $\widehat{E}(p)$ be a positive valued, continuously
    differentiable function of~$p$ bounding $|\pi(p)-~li(p)|$ on
    $[\intrangeleft,\intrangeright]$.  (For example, under the Riemann
    hypothesis we can take $\widehat{E}(p) = \frac{1}{8\pi} \sqrt{p}
    ~ln p$ provided $B \geq 2657$.)  Then
    \begin{equation*}
      \Biggl|
      \primesum{f(p)}{p}{\intrangeleft}{\intrangeright}
      -
      \quad\widetilde{g}\quad
      \Biggr|
      \;\leq\;
      \widehat{g}
    \end{equation*}
    with
    \begin{align*}
      \widetilde{g} &= \frint{ \frac{\widetilde{f}(p)}{~ln p}
      }{p}{\intrangeleft}{\intrangeright},\\
      \widehat{g} &= \frint{ \frac{\widehat{f}(p)}{~ln p}
      }{p}{\intrangeleft}{\intrangeright}
      + 2 (\widetilde{f}+\widehat{f})(\intrangeleft)
      \widehat{E}(\intrangeleft) + 2
      (\widetilde{f}+\widehat{f})(\intrangeright)
      \widehat{E}(\intrangeright)
      + \frint{ \left( \widetilde{f} + \widehat{f} \right)(p)
        \widehat{E}'(p) }{p}{\intrangeleft}{\intrangeright}.
    \end{align*}
    In the special case when $\widetilde{f}+\widehat{f}$ is constant
    we have the better bound
    \begin{align*}
      \widehat{g} &= \frint{ \frac{\widehat{f}(p)}{~ln p}
      }{p}{\intrangeleft}{\intrangeright}
      + (\widetilde{f}+\widehat{f})(\intrangeleft)
      (\widehat{E}(\intrangeleft) + \widehat{E}(\intrangeright))
    \end{align*}
  \end{lemma}
  \begin{proof}
    The proof can be done analogously to the proof of Lemma
    2.1 
    in \ifanonymous Anonymous (2010)\else\cite{loenus10c}\fi:
    First, rewrite $\primesum{f(p)}{p}{\intrangeleft}{\intrangeright}$
    as a Stieltjes integral
    $\frint{f(p)}{p}{\intrangeleft}{\intrangeright}$.  Then integrate
    by parts, estimate $\pi$, and finally integrate by parts
    `backwards'.
  \end{proof}

  Next we formulate a lemma specialized to handle RSA notions.  We
  cannot expect to obtain an approximation of the number of prime
  pairs by the area of the region unless we make certain restrictions.

\end{fullversion}%
\begin{shortversion}%
  We skip intermediate results here and just summarize the number
  theoretic work (to ease later comparison we have retained the
  numbering of the extended version\ifanonymous\else\
  \citealt{loenus11a}\fi).
\end{shortversion}%
\begin{fullversion}[\refstepcounter{equation}]
  The following definition describes the restrictions that we use.
  As you will notice, it essentially enforces a certain monotonicity
  that allows the error estimation.
  \begin{definition}\label{def:notion-special}
    Let $\RSASET$ be a notion of RSA integers with tolerance~$r$.
    \begin{enumerate}
    \item\label{def:notion-graph-bounded} The notion $\RSASET$ is
      \emph{graph-bounded} iff there are (at least) integrable
      boundary functions
      $\shortmap[B_{1}, C_{1}]{\R_{>1}}{\R_{>1}}$ and
      $\shortmap[B_{2}, C_{2}]{\R_{>1}^{2}}{\R_{>1}}$ such that we can
      write
      $$
      \RSASET_{x} = \Set{ (y,z) \in \R_{>1}^{2};
        \begin{array}{c}
          B_{1}(x) < y \leq C_{1}(x),\\
          B_{2}(y,x) < z \leq C_{2}(y,
          x)
        \end{array}},
      $$
      where for all $x \in \R_{>1}$ and all $y \in
      \left]B_{1}(x),C_{1}(x)\right[$ we have $1 < B_{1}(x) \leq
      C_{1}(x) \leq x$ and $1 < B_{2}(y,x) < C_{2}(y,x) \leq x$.
    \item\label{def:notion-monotone} The notion $\RSASET$ is
      \emph{monotone at $x$ (relative to the error bound
        $\widehat{E}$)} for some $x\in\R_{>1}$ iff it is graph-bounded
      and the function
      $$
      \frint{\frac{1}{~ln q}}{q}{B_{2}(p,x)}{C_{2}(p,x)} +
      \widehat{E}(B_{2}(p,x)) + \widehat{E}(C_{2}(p,x))
      $$
      is either weakly increasing, weakly decreasing, or constant as a
      function in $p$ restricted to the interval $[B_{1}(x),
      C_{1}(x)]$.
      If not mentioned otherwise we refer to the error bound given by
      $\widehat{E}(p) = \frac{1}{8\pi} \sqrt{p} ~ln p$.

      We call the notion $\RSASET$ \emph{monotone} iff it is monotone
      at each $x \in \R_{>1}$ where $\RSASET_{x} \ne \emptyset$.

      \begin{fullversion}%
      \item\label{def:notion-piecewise-monotone} The notion $\RSASET$
        is \emph{piecewise monotone} iff there is a parameter $m \in
        \N$ such that
        $$
        \RSASET_{x} := \biguplus_{j = 1}^{m} \RSASET_{j,x},
        $$
        where $\RSASET_{j,.}$ are all monotone notions of RSA integers
        of tolerance~$r$. Note that we may also allow $m$ to depend on
        $x$.
      \end{fullversion}%
    \end{enumerate}
  \end{definition}
  For \ref{def:notion-graph-bounded} note that $B_{1}(x) = C_{1}(x)$
  allows to describe an empty set $\RSASET_{x}$, and otherwise the
  inequality $B_{2}(y,x) \ne C_{2}(y,x)$ makes sure that all four
  bounding functions are determined by $\RSASET_{x}$ as long as $y \in
  \left] B_{1}(x), C_{1}(x) \right[$.  This condition enforces that
  $\RSASET_{x}$ is (path) connected.  We do not need that but also it
  does no harm.
  For \ref{def:notion-piecewise-monotone} observe that in the light of
  a multi-application of \ref{res:psalm23} we would be on the safe
  side if we require $m \in ~ln^{\bigO{1}} x$. At the extreme $m \in
  \smallo{c_1 x^{\frac{1-c}{4}} ~ln x}$ with $c = ~max\left( 2c_{2}-1,
    1-2c_{1} \right)$ is necessary for any meaningful result
  generalizing \ref{res:psalm23}.
\end{fullversion}%
\begin{fullversion}[\refstepcounter{equation}]%
  As in particular \ref{def:notion-monotone} is rather weird to verify
  we provide an easily checkable, sufficient condition for
  monotonicity of a notion.
  \begin{lemma}
    \label{res:criterionMonotone}
    Assume $\RSASET$ is a graph-bounded notion of RSA integers with
    tolerance~$r$ given by continuously differentiable functions
    $\shortmap[B_{1}, C_{1}]{\R_{>1}}{\R_{>1}}$ and $\shortmap[B_{2},
    C_{2}]{\R_{>1}^{2}}{\R_{>1}}$.  Finally, let $x \in \R_{>1}$ be
    such that
    \begin{itemize}
    \item the function $B_{2}(p,x)$ is weakly decreasing in $p$ and
    \item the function $C_{2}(p,x)$ is weakly increasing in $p$
    \end{itemize}
    for $p \in \left]B_{1}(x), C_{1}(x)\right]$, or vice versa.  As
    usual let $\widehat{E}(p)$ be the function given by
    $\widehat{E}(p) = \frac{1}{8\pi} \sqrt{p} ~ln p$. Then the notion
    $\RSASET$ is monotone at $x$ (relative to $\widehat{E}$).
  \end{lemma}
  \begin{fullversion}
    \begin{proof}
      The goal is to show that the function
      $$
      h(p) := \frint{\frac{1}{~ln q}}{q}{B_{2}(p,x)}{C_{2}(p,x)} +
      \widehat{E}(B_{2}(p,x)) + \widehat{E}(C_{2}(p,x))
      $$
      is weakly increasing or weakly decreasing in $p$. We write
      $B_{2}'(p,x)$ and $C_{2}'(p,x)$, respectively, for the
      derivative with respect to $p$. Note that
      \begin{align*}
        h'(p) &:=
        \begin{aligned}[t]
          &\underbrace{ \left( \frac{1}{~ln C_{2}(p,x)} + \frac {2+~ln
                C_{2}(p,x)} {16 \pi \sqrt{C_{2}(p,x)}} \right) }_{>0}
          C_{2}'(p,x)
          \\
          & - \underbrace{\left( \frac{1}{~ln B_{2}(p,x)} -
              \frac{2+~ln B_{2}(p,x)}{16 \pi \sqrt{B_{2}(p,x)}}
            \right)}_{>0} B_{2}'(p,x).
        \end{aligned}
      \end{align*}
      Some simple calculus shows that the second underbraced term is always
      positive since $B_{2}(p,x) > 1$.
      Thus if $B_{2}(p,x)$ is weakly decreasing and $C_{2}(p,x)$ is
      weakly increasing, we have that $h(p)$ is weakly increasing.  If
      on the other hand $B_{2}(p,x)$ is weakly increasing and
      $C_{2}(p,x)$ is weakly decreasing it follows that $h(p)$ is
      weakly decreasing.
    \end{proof}
  \end{fullversion}
  Clearly, the conditions of the lemma are not necessary.
\end{fullversion}%
\begin{fullversion}[\refstepcounter{equation}]%
  We can easily extended it, for example, as follows:
  \begin{lemma}
    \label{res:criterionMonotoneSplit}
    Assume $\RSASET$ is a graph-bounded notion of RSA integers with
    tolerance~$r$ given by continuously differentiable functions
    $\shortmap[B_{1}, C_{1}]{\R_{>1}}{\R_{>1}}$ and $\shortmap[B_{2},
    C_{2}]{\R_{>1}^{2}}{\R_{>1}}$.  Further, individually for each $x
    \in \R_{>1}$, the functions $B_{2}(p,x)$ and $C_{2}(p,x)$ are both
    weakly increasing in $p$ for $p \in \left]B_{1}(x),
      C_{1}(x)\right]$.
    Then there are two monotone notions $\RSASET^{1}$ and
    $\RSASET^{2}$ with tolerance~$r$, both having $\RSASET^{i}_{x}
    \subseteq \R_{\geq B_{1}(x)} \times \R_{\geq B_{2}(B_{1}(x),x)}$
    for all $x$, such that $\RSASET = \RSASET^{1} \setminus
    \RSASET^{2}$.
  \end{lemma}
  \begin{proof}
    Let $A(x) := B_{2}( B_{1}(x), x )$.  We define two $[c_{1},
    c_{2}]$-balanced graph-bounded notions $\RSASET^{1}$,
    $\RSASET^{2}$ of RSA integers by the following: the first notion
    $\RSASET^{1}$ is defined by the functions $B_{1}^{1} := B_{1}$,
    $C_{1}^{1} := C_{1}$, $B_{2}^{1}(p,x) := A(x)$ and $C_{2}^{1} :=
    C_{2}$.  The second notion $\RSASET^{2}$ is defined by the
    functions $B_{1}^{1} := B_{1}$, $C_{1}^{1} := C_{1}$,
    $B_{2}^{2}(p,x) := A(x)$ and $C_{2}^{2} := B_{2}$.  Since $x/r <
    B_{1}(x) B_{2}(B_{1}(x),x) = B_{1}(x) A(x)$ both new notions have
    tolerance~$r$ as well.  Then $\RSASET^{1}$, $\RSASET^{2}$ are by
    \ref{res:criterionMonotone} both monotone and $\RSASET =
    \RSASET^{1} \setminus \RSASET^{2}$.
  \end{proof}

  A similar result with $B_{2}$ and $C_{2}$ both weakly decreasing is
  more difficult to obtain while simultaneously retaining the
  tolerance.  A particularly difficult example is the maximal notion
  $\RSAMAXSET$ given by $\RSAMAXSET_{x} = \Set{(y,z) \in \R_{>1}^{2};
    \frac{x}{r} < y z \leq x \land y,z \geq x^{c_{1}}}$.
\end{fullversion}%
The following lemma covers all the estimation work.
\begin{fullversion}
  Notice that we could in
  principle obtain explicit values for the $\bigO{}$ constant
  based on \ref{res:recapprox} but the expressions are rather ugly.
\end{fullversion}
\vspace{-0.2cm}
\begin{lemma}[Two-dimensional prime sum approximation for monotone
  notions]
  \label{res:psalm23}
  \hskip2em\hskip0pt\linebreak[3]\hskip-2em plus 1em%
  Assume that we have a monotone $[c_{1},c_{2}]$-balanced notion
  $\RSASET$ of RSA integers with tolerance~$r$, where $0 < c_{1} \leq
  c_{2}$.  (The values $r$, $c_{1}$, $c_{2}$ are allowed to vary with
  $x$.)  Then under the Riemann hypothesis there is a value
  $\widetilde{a}(x) \in \left[ \frac{1}{4c_{2}^{2}},
    \frac{1}{4c_{1}^{2}} \right]$ such that
  $$
  \RSA(x) \in%
  \widetilde{a}(x) \cdot \frac{4 ~area(\RSASET_{x})}{~ln^{2} x} +
  \bigO{
    {c_{1}^{-1}} {x^{\frac{3+c}{4}}} },
  $$
  where $c = ~max\left( 2c_{2}-1,
    1-2c_{1} \right)$. \ifshortversion\qed\fi
\end{lemma}
Note that the
\begin{shortversion}[following ]%
  omitted%
\end{shortversion}
proof gives a precise expression for $\widetilde{a}(x)$, namely
$$
\widetilde{a}(x)
= \frac{ \iint_{\RSASET_{x}} \frac{1}{~ln p ~ln q} \differential{p}
  \differential{q} }{ 4 \iint_{\RSASET_{x}} \frac{1}{~ln^{2} x}
  \differential{p} \differential{q} } .
$$
It turns out that we can only evaluate $\widetilde{a}(x)$ numerically
in our case and so we tend to estimate also this term.  Then we often
obtain $\widetilde{a}(x) \in 1+o(1)$.  Admittedly, this mostly eats up
the advantage obtained by using the Riemann hypothesis.  However, we
accept this because it still leaves the option of going through that
difficult evaluation and obtain a much more precise answer.  If we do
not use the Riemann hypothesis we need to replace $\bigO{c_{1}^{-1}
  x^{\frac{3+c}{4}}}$ with $\bigO{\frac{x}{~ln^{k} x}}$ for any $k>2$
of your choice.

\begin{fullversion}[\refstepcounter{equation}\refstepcounter{equation}\refstepcounter{equation}\refstepcounter{equation}]
  \begin{proof}
    Fix any $x\in\R_{>1}$.  In case $~area( \RSASET_{x} ) = 0$ the
    claim holds with any desired $\widetilde{a}(x)$ and zero big-Oh
    term.  We can thus assume that the area is positive.  As the
    statement is asymptotic and $x^{c_{1}}$ tends to $\infty$ with $x$
    we can further assume that $x^{c_{1}} \geq 2657$.
    Abbreviating $\widetilde{h}(x) = \frac{4 ~area(\RSASET_{x})}
    {~ln^{2} x} $, we prove that there exists a value
    $\widetilde{a}(x) \in \left[ \frac{1}{4 c_{2}^{2}}, \frac{1}{4
        c_{1}^{2}} \right]$ such that
    $$
    \left| \RSA(x) \quad-\quad \widetilde{a}(x) \cdot \widetilde{h}(x)
    \right| \leq \widehat{h}(x)
    $$
    with
    \begin{align*}
      \widehat{h}(x) &=
      \frac{1}{4\pi c_{1}} \left( 7-6c_{2}+\frac{12}{~ln x} \right)
      x^{\frac{1+c_{2}}{2}} + \frac{1}{8\pi^{2}} \cdot
      x^{\frac{1}{2}+\frac{2 ~ln~ln x}{~ln x}} + \frac{1}{4\pi c_{1}}
      \left( 1+\frac{4}{~ln x} \right) x^{1-\frac{c_{1}}{2}}.
    \end{align*}
    This is slightly more precise and implies the claim.

    Since the given notion is $[c_{1},c_{2}]$-balanced with
    tolerance~$r$ for any $(y,z) \in \RSASET_{x}$ we have $\frac{x}{r}
    \leq y z \leq x$ and $y, z \in [x^{c_{1}}, x^{c_{2}}]$ which
    implies $~ln y, ~ln z \in [c_{1}, c_{2}] ~ln x$.
    Equivalently, we have
    \begin{align}
      \label{eq:B1C1balance}
      x^{c_{1}} \leq B_{1}(x) &\leq C_{1}(x) \leq x^{c_{2}}
    \end{align}
    and for $y \in \left]B_{1}(x), C_{1}(x)\right[$ we have
    \begin{align}
      \label{eq:BCtolerance}
      \frac{x}{r y} \leq B_{2}(y,x) < C_{2}(y, x) \leq \frac{x}{y}
    \end{align}
    and
    \begin{align}
      \label{eq:B2C2balance}
      x^{c_{1}} \leq B_{2}(y,x) < C_{2}(y, x) \leq x^{c_{2}}.
    \end{align}
    From \ref{eq:BCtolerance} we infer that for all $y \in
    \left]B_{1}(x), C_{1}(x)\right[$ we have
    \begin{align}
      \label{eq:BCproducts}
      \frac{x}{r} \leq y B_{2}(y,x) \leq
      x\quad\text{and}\quad\frac{x}{r} \leq y C_{2}(y,x) \leq x.
    \end{align}
    In order to estimate
    $$
    \RSA(x) \quad=\qquad \primesum{\quad \primesum{\quad
        1}{q}{B_{2}(p,x)}{C_{2}(p,x)}}{p}{B_{1}(x)}{C_{1}(x)},
    $$
    we apply \ref{res:recapprox} twice.  Since $x^{c_{1}} \geq 2657$
    and so $B_{2}(p,x) \geq 2657$ for the considered $p$ we obtain for
    the inner sum
    \begin{equation*}
      \Biggl|
      \primesum{1}{q}{B_{2}(p,x)}{C_{2}(p,x)}\quad
      -
      \quad\widetilde{g}_{1}(p,x)\quad
      \Biggr|
      \;\leq\;
      \widehat{g}_{1}(p,x),
    \end{equation*}
    where
    \begin{align*}
      \widetilde{g}_{1}(p,x) &= \frint{ \frac{1}{~ln q}
      }{q}{B_{2}(p,x)}{C_{2}(p,x)},
      \\
      \widehat{g}_{1}(p,x) &= \widehat{E}(B_{2}(p,x)) +
      \widehat{E}(C_{2}(p,x)),
    \end{align*}
    since we can use the special case of constant functions in
    \ref{res:recapprox}.  Because we are working under the restriction
    that the notion is monotone, i.e.\ $\widetilde{g}_{1}(p,x) +
    \widehat{g}_{1}(p,x)$ is monotone, we are able to apply the lemma
    a second time.  Since $x^{c_{1}} \geq 2657$ and so $B_{1}(x) \geq
    2657$ we obtain
    \begin{equation*}
      \Biggl|
      \primesum{
        \primesum{1}{q}{B_{2}(p,x)}{C_{2}(p,x)}
      }{p}{B_{1}(x)}{C_{1}(x)}\quad
      -
      \quad\widetilde{g}_{2}(x)\quad
      \Biggr|
      \;\leq\;
      \widehat{g}_{2}(x),
    \end{equation*}
    where
    \begin{footnotesize}
      \begin{align*}
        \widetilde{g}_{2}(x) &= \frint{\frint{\frac{1}{~ln p ~ln
              q}}{q}{B_{2}(p,x)}{C_{2}(p,x)}}{p}{B_{1}(x)}{C_{1}(x)},
        \\
        \widehat{g}_{2}(x) &=
        \begin{aligned}[t]
          & \frac{1}{8 \pi} \preint{p}{B_{1}(x)}{C_{1}(x)}
          \left(\sqrt{B_{2}(p,x)} ~ln B_{2}(p,x) + \sqrt{C_{2}(p,x)}
            ~ln C_{2}(p,x) \right) \cdot \left( \frac{1}{~ln p} +
            \frac{~ln p + 2}{2 \sqrt{p}} \right) \postint
          \\
          &+ \frac{1}{4 \pi} \sqrt{B_{1}(x)} ~ln B_{1}(x)
          \preint{q}{B_{2}(B_{1}(x),x)}{C_{2}(B_{1}(x),x)}
          \frac{1}{~ln q} \postint
          \\
          &+ \frac{1}{4 \pi} \sqrt{C_{1}(x)} ~ln C_{1}(x)
          \preint{q}{B_{2}(C_{1}(x),x)}{C_{2}(C_{1}(x),x)}
          \frac{1}{~ln q} \postint
          \\
          &+ \frac{1}{32 \pi^{2}} \sqrt{B_{1}(x)} ~ln B_{1}(x) \left(
            \sqrt{B_{2}(B_{1}(x),x)} ~ln
            \left(B_{2}(B_{1}(x),x)\right) + \sqrt{C_{2}(B_{1}(x),x)}
            ~ln \left(C_{2}(B_{1}(x),x)\right) \right)
          \\
          &+ \frac{1}{32 \pi^{2}} \sqrt{C_{1}(x)} ~ln C_{1}(x) \left(
            \sqrt{B_{2}(C_{1}(x),x)} ~ln
            \left(B_{2}(C_{1}(x),x)\right) + \sqrt{C_{2}(C_{1}(x),x)}
            ~ln \left(C_{2}(C_{1}(x),x)\right) \right)
          \\
          &+ \frac{1}{8 \pi} \preint{p}{B_{1}(x)}{C_{1}(x)}
          \preint{q}{B_{2}(p,x)}{C_{2}(p,x)} \frac{~ln p+2}{2\sqrt{p}
            ~ln q} \postint \postint.
        \end{aligned}
      \end{align*}
    \end{footnotesize}
    It remains to estimate $\widetilde{g}_{2}(x)$ and
    $\widehat{g}_{2}(x)$ suitably sharply.

    For $(p,q) \in \RSASET_{x}$ we frequently use the estimate $~ln p,
    ~ln q \in [c_{1}, c_{2}] ~ln x$.  For the main term we obtain
    $$
    \widetilde{g}_{2}(x) \in \left[ \frac{1}{4 c_{2}^{2}}, \frac{1}{4
        c_{1}^{2}}\right] \frac{4 ~area(\RSASET_{x})}{~ln^{2} x} .
    $$
    We also read off the exact expression $\widetilde{a}(x) =
    \frac{~ln^{2}x}{4 ~area(\RSASET_{x})} \widetilde{g}_{2}(x)$.
    
    \noindent We treat the error term $\widehat{g}_{2}(x)$ part by part.  For
    the first term we obtain
    \begin{align*}
      \frac{1}{8 \pi} \preint{p}{B_{1}(x)}{C_{1}(x)}
      &\left(\sqrt{B_{2}(p,x)} ~ln B_{2}(p,x) + \sqrt{C_{2}(p,x)} ~ln
        C_{2}(p,x) \right) \cdot \left( \frac{1}{~ln p} + \frac{~ln p
          + 2}{2 \sqrt{p}} \right) \postint
      \\
      &\leq \frac{1}{4 \pi} \preint{p}{x^{c_{1}}}{x^{c_{2}}}
      \sqrt{\frac{x}{p}} ~ln\left(\frac{x}{p}\right) \cdot
      \frac{3}{~ln p} \postint
      \\
      &\leq \frac{3}{4\pi} \frac{1}{c_{1} ~ln x}
      \preint{p}{x^{c_{1}}}{x^{c_{2}}} \sqrt{\frac{x}{p}}
      ~ln\left(\frac{x}{p}\right) \postint
      \\
      &\leq \frac{3}{2\pi} \frac{1}{c_{1}} \left( 1-c_{2} +
        \frac{2}{~ln x} \right) x^{\frac{1+c_{2}}{2}}
      \in \bigO{c_1^{-1} x^{\frac{1+c_{2}}{2}} },
    \end{align*}
    where we used in the second line that $\frac{~ln p + 2}{2
      \sqrt{p}} \leq \frac{2}{~ln p}$ for all $p \geq 2$. Basic
    calculus shows that $\frac{~ln p (~ln p + 2)}{2 \sqrt{p}}$ is
    maximal at $p=~exp(\sqrt{5}+1)$, where it is less than $1.68$.
    For the fourth line note that
    $$
    \frint{\sqrt{\frac{x}{p}} ~ln\left(\frac{x}{p}\right)}{p}{}{} = 2p
    \sqrt{\frac{x}{p}} \left(~ln\left(\frac{x}{p}\right) + 2\right).
    $$
    The definite integral is not greater than this function evaluated
    at $p = x^{c_{2}}$ since $c_{1} \leq \frac{1}{2}$. Using $c_{2}
    \geq 0$ gives the claim.

    The second term yields
    \begin{align*}
      \frac{1}{8 \pi} \sqrt{B_{1}(x)} & ~ln B_{1}(x)
      \preint{q}{B_{2}(B_{1}(x),x)}{C_{2}(B_{1}(x),x)} \frac{1}{~ln q}
      \postint
      \\
      &\leq \frac{1}{8 \pi c_{1} ~ln x} \sqrt{B_{1}(x)}
      C_{2}(B_{1}(x),x) ~ln B_{1}(x)
      \\
      &\leq \frac{1}{8 \pi c_{1}} x^{\frac{1+c_{2}}{2}}
      \in \bigO{c_1^{-1} x^{\frac{1+c_{2}}{2}}},
    \end{align*}
    since 
    we have $\sqrt{B_{1}(x) C_{2}(B_{1}(x),x)}
    \sqrt{C_{2}(B_{1}(x),x)} \leq x^{\frac{1+c_{2}}{2}}$ and $~ln
    B_{1}(x) \leq ~ln x$.

    Similarly we obtain for the third term
    \begin{align*}
      \frac{1}{8 \pi} \sqrt{C_{1}(x)} & ~ln C_{1}(x)
      \preint{q}{B_{2}(C_{1}(x),x)}{C_{2}(C_{1}(x),x)} \frac{1}{~ln q}
      \postint
      \\
      &\leq \frac{1}{8\pi c_{1}} x^{\frac{1+c_{2}}{2}} \in \bigO{
        c_1^{-1} x^{\frac{1+c_{2}}{2}} },
    \end{align*}
    using $\sqrt{C_{1}(x) C_{2}(C_{1}(x),x)} \sqrt{C_{2}(C_{1}(x),x)}
    \leq x^{\frac{1+c_{2}}{2}}$ and $~ln C_{1}(x) \leq ~ln x$.
    
    The fourth term yields
    \begin{align*}
      \frac{1}{32 \pi^{2}} \sqrt{B_{1}(x)} ~ln B_{1}(x) & \left(
        \sqrt{B_{2}(B_{1}(x),x)} ~ln B_{2}(B_{1}(x),x) +
        \sqrt{C_{2}(B_{1}(x),x)} ~ln C_{2}(B_{1}(x),x) \right)
      \\
      &\leq \frac{1}{16 \pi^{2}} \sqrt{x} ~ln^{2} x
      \in \bigO{ x^{\frac{1+c_{2}}{2}} },
    \end{align*}
    where we used \ref{eq:BCproducts} and the (very weak) bound $~ln
    B_{1}(x), ~ln C_{2}(p,x) \leq ~ln x$.
    The fifth term can be treated similarly.
    We finish by observing for the sixth term
    \begin{align*}
      \frac{1}{8 \pi} \preint{p}{B_{1}(x)}{C_{1}(x)}
      \preint{q}{B_{2}(p,x)}{C_{2}(p,x)} &\frac{~ln p+2}{2\sqrt{p} ~ln
        q} \postint \postint
      \\
      &\leq \frac{1}{8 \pi} \frac{1}{c_{1} ~ln x}
      \preint{p}{B_{1}(x)}{C_{1}(x)}
      \preint{q}{B_{2}(p,x)}{C_{2}(p,x)} \frac{~ln p}{\sqrt{p}}
      \postint \postint
      \\
      &\leq \frac{1}{8 \pi} \frac{1}{c_{1} ~ln x}
      \preint{p}{x^{c_{1}}}{x^{c_{2}}} \frac{~ln p}{\sqrt{p}}
      \preint{q}{0}{\frac{x}{p}}{} \postint \postint
      \\
      &\leq \frac{1}{8 \pi} \frac{1}{c_{1} ~ln x} \cdot x \cdot
      \preint{p}{x^{c_{1}}}{x^{c_{2}}} \frac{~ln p}{p^{3/2}} \postint
      \\
      &\leq \frac{1}{4 \pi} \frac{1}{c_{1}} \left( 1 + \frac{4}{ ~ln
          x} \right) x^{1-\frac{c_{1}}{2}}
      \\
      & \in \bigO{ c_1^{-1} x^{1-\frac{c_{1}}{2}} }
    \end{align*}
    using $B_{1}(x) \geq x^{c_{1}}$, $c_{1} \leq \frac{1}{2}$, and
    $$\frint{\frac{~ln p}{p^{3/2}}}{p}{}{} = \frac{-2 (~ln p + 2)}{\sqrt{p}}.$$

    This completes the proof.
  \end{proof}
  In specific situations one may obtain better estimates.  In
  particular, when we substitute $C_{2}(p,x)$ by $x/p$ in the
  estimation of the sixth summand of the error we may loose much.

  Of course we can generalize this lemma to notions composed of few
  monotone ones.  We leave the details to the reader.
\end{fullversion}
As mentioned before, in many standards the selection of the primes $p$
and $q$ is additionally subject to the side condition that
$~gcd((p-1)(q-1), e) = 1$ for some fixed public exponent $e$ of the
RSA cryptosystem. To handle these restrictions,
we \begin{shortversion}[prove]state a theorem from the extended
  version \end{shortversion}
\begin{theorem}\label{res:coprimalityE}
  Let $e \in \N_{> 2}$ be a public RSA exponent and $x \in \R$. Then
  under the Extended Riemann Hypothesis we have for the number
  $\pi_{e}(x)$ of primes $p \leq x$ with $~gcd(p-1,e) = 1$ that
  $$
  \pi_{e}(x) \in \frac{\phi_1(e)}{\phi(e)} \cdot ~li(x) +
  \bigO{\sqrt{x} ~ln x},
  $$
\end{theorem}
where $~li(x) = \preint{t}{0}{x} \frac{1}{~ln t} \postint$ is the
integral logarithm, $\phi(e)$ is Euler's totient function and
\begin{equation}\label{def:phi1}
  \frac{\phi_1(e)}{\phi(e)} = \prod_{\substack{l \mid e\\l \text{ prime}}} \left(1 - \frac{1}{l-1}\right).
\end{equation}
\ifshortversion\expandafter\qed\fi
\begin{fullversion}
  \begin{proof}
    We first show that the number of elements in $\Z_e^\times \cap
    (1+\Z_e^\times)$ is exactly $\phi_1(e)$.  Write $e =
    \prod_{\substack{l \mid e\\ l \text{ prime}}} l^{f(l)}$. Observe
    that by the Chinese Remainder Theorem we have
    $$
    \Z_e^\times \cap (1+\Z_e^\times) = \bigoplus_{\substack{l \mid
        e\\l \text{ prime}}} \left( \Z_{l^{f(l)}}^\times \cap (1 +
      \Z_{l^{f(l)}}^\times) \right)
    $$
    and each factor in this expression has size $(l-2) l^{f(l) -
      1}$. Multiplying up all factors gives
    $$
    \#(\Z_e^\times \cap (1+\Z_e^\times)) = \prod_{\substack{l \mid
        e\\l \text{ prime}}} \left(1 - \frac{1}{l-1}\right) \left(1 -
      \frac{1}{l}\right) l^{f(l)} = \phi_1(e).
    $$
    To show the result for $\pi_{e}(x)$ 
    note that \citet{oes79} implies the following quantitative version
    of \citeauthor{dir37}'s theorem on the number $\pi_{e; a}(x)$ of
    primes $p \leq x$ with $p \equiv a \in \Z_e$ when $~gcd(a,e) = 1$
    under the Extended Riemann Hypothesis
    $$
    \left| \pi_{e; a}(x) - \frac{1}{\phi(e)} \cdot ~li(x) \right| \leq \sqrt{x}
      (~ln x + 2~ln e).
    $$
    This is also documented in \citet[Theorem 8.8.18]{bacsha96}.  We
    now have to sum over $\phi_1(e)$ residue classes and so obtain
    $$
    \pi_{e}(x) \in \frac{\phi_1(e)}{\phi(e)} \cdot ~li(x) +
    \bigO{\phi_1(e) \sqrt{x} ~ln x},
    $$
    which proves the claim.
  \end{proof}
\end{fullversion}%
This theorem shows that the prime pair approximation in
\ref{res:psalm23} can be easily adapted to RSA integers whose prime
factors satisfy the conditions of \ref{res:coprimalityE} (when
assuming the Extended Riemann Hypothesis), since the density of such
primes differs for every fixed $e$ essentially just by a
multiplicative constant compared to the density of arbitrary primes.

\section{Some common definitions for RSA integers}
\label{sec:NotionsInParticular}

We will now give formal definitions of \begin{shortversion}[three
  ]two\end{shortversion} specific notions of RSA integers.  In
particular, we consider the following example definitions within our
framework:
\begin{itemize}
  \begin{fullversion}
  \item The number theoretically inspired notion following
    \citeauthor{decmor08}.  Note that this occurs in no standard and
    no implementation.
  \end{fullversion}
\item The simple construction given by just choosing two primes in
  given intervals. This construction occurs in several standards, like
  the standard of the RSA foundation \citep{rsa00}, the standard
  resulting from the European NESSIE project \citep{prebir03} and the
  FIPS 186-3 standard \citep{fips1863}.  Also open source
  implementations of \texttt{OpenSSL} \citep{openssl09},
  \texttt{GnuPG} \citep{gnupg09} and the GNU crypto library
  \texttt{GNU Crypto} \citep{gnucrypto09} use some variant of this
  construction.

\item An algorithmically inspired construction which allows one prime
  being chosen arbitrarily and the second is chosen such that the
  product is in the desired interval.  This was for example specified
  as the IEEE standard 1363 \citep{araarn00}, Annex A.16.11.  However,
  we could not find any implementation following this standard.
\end{itemize}
  \begin{fullversion}%
    \begin{figure}[h]
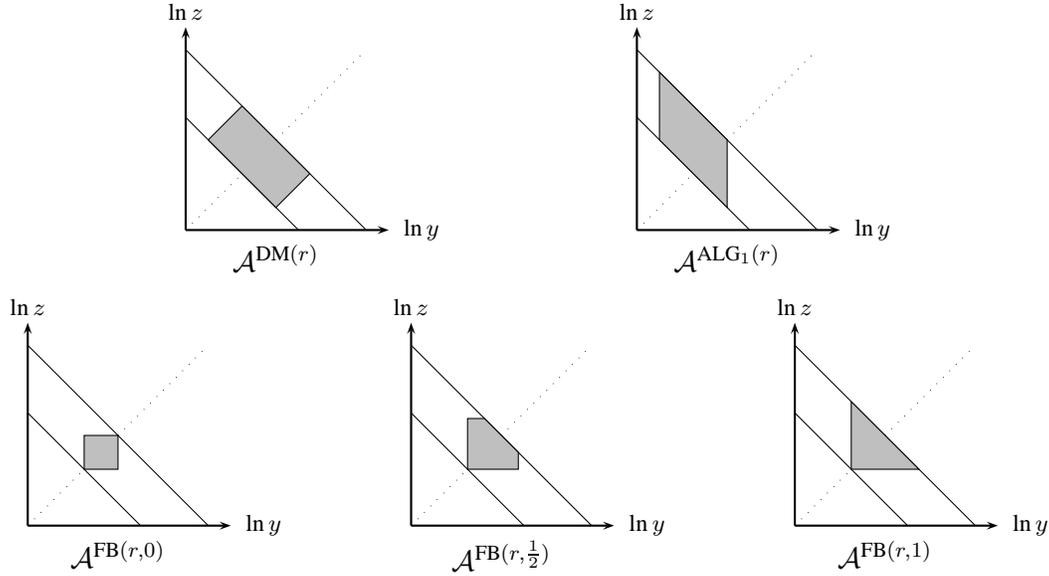

      \centering \psset{unit=3mm}%
      \newpsstyle{notionsplit}{linestyle=none}%
      \hfill%
      \drawnotion[{%
        \DNtitle{$\RSADMSET$}%
      }]{8}{3}{\RSADMdraw}%
      \hfill%
      \drawnotion[{%
        \DNtitle{$\RSAALGSET$}%
      }]{8}{3}{\RSAALGdraw}%
      \hfill%
      \null\\
      \drawnotion[{%
        \DNtitle{$\RSAFIXSET[r,0]$}%
      }]{8}{3}{\RSAFIXdraw{0}}%
      \hfill%
      \drawnotion[{%
        \DNtitle{$\RSAFIXSET[r,\frac{1}{2}]$}%
      }]{8}{3}{\RSAFIXdraw{0.5}}%
      \hfill%
      \drawnotion[{%
        \DNtitle{$\RSAFIXSET[r,1]$}%
      }]{8}{3}{\RSAFIXdraw{1}}%
      \caption{Three notions of RSA integers.}
      \label{fig:RSAThreeNotions}
    \end{figure}
  \end{fullversion}%

\begin{fullversion}[\refstepcounter{equation}\refstepcounter{equation}\refstepcounter{equation}]
  \subsection{A number theoretically inspired notion}

  In \cite{decmor08} on the suggestion of B.~de~Weger, the number
  $\QDM(x)$ of RSA~integers up to~$x$ was defined as the count of
  numbers whose two prime factors differ by at most a factor $r$,
  namely
  \begin{equation*}
    \QDM(x)
    :=
    \#\Set{
      n \in \N;
      \begin{array}{l}
        \exists p, q \in \mathbb{P} \colon\\
        n=p q \Land p < q < r p \Land n \leq x 
      \end{array}
    }.
  \end{equation*}
  \condbreak{10mm}\noindent%
  Written as a notion of RSA integers in the sense above, we analyze
  \begin{fullversion}
    \hanghere[1.3\baselineskip]{\psset{unit=1.5mm}\drawnotion[]{8}{3}{\RSADMdraw}}
  \end{fullversion}
  \begin{equation}
    \label{def:RSA1}
    \RSADMSET
    := \left\langle \Set{ (y,z) \in \R^{2}; \frac{y}{r} < z < r y 
        \Land
        \frac{x}{r} < y z \leq x } \right\rangle_{x \in \R_{>1}}.
  \end{equation}
  Note that the prime pair counting function of this notion is closely
  related to the function $\QDM(x)$: Namely we have
  $$
  \RSADM(x) = 2\left(\QDM(x)-\QDM(\frac{x}{r})\right) + \left(
    \pi\left(\sqrt{x}\right) - \pi\left( \sqrt{\frac{x}{r}} \right)
  \right),
  $$
  where the last part is comparatively small.
  \begin{fullversion}[We obtain using
    \ref{res:psalm23}:\refstepcounter{equation}]%
    We now analyze the behavior of the function $\RSADM(x)$ under the
    Riemann hypothesis.  Similar to \cite{decmor08}, we rewrite
    \begin{equation}
      \label{eq:RSA1SUM}
      \frac{1}{2} \cdot   \RSADM(x)
      = 
      \begin{aligned}[t]
        & \primesum{ \primesum{1}{q}{\frac{x}{r p}}{r p}
        }{p}{\frac{\sqrt{x}}{r}}{\sqrt{\frac{x}{r}}}\\
        &+ \primesum{ \primesum{1}{q}{p}{\frac{x}{p}}
        }{p}{\sqrt{\frac{x}{r}}}{\sqrt{x}} \quad+\quad
        \frac{\pi\left(\sqrt{x}\right) - \pi\left( \sqrt{\frac{x}{r}}
          \right)}{2}.
      \end{aligned}
    \end{equation}
    With these bounds we obtain using \ref{res:psalm23}:
  \end{fullversion}
  \begin{theorem}\label{res:order-rsa1}
    Under the Riemann hypothesis we have
    $$
    \RSADM(x) \in \widetilde{a}(x) \frac{4 x}{~ln^{2} x} \left(~ln r -
      \frac{~ln r}{r} \right) + \bigO{ x^{\frac{3}{4}} r^{\frac{1}{2}}
    }
    $$
    with $\widetilde{a}(x) \in \left[ \left(1 - \frac{~ln r}{~ln x +
          ~ln r}\right)^2, \left(1 + \frac{2 ~ln r}{~ln x - 2 ~ln
          r}\right)^2\right]$.  This makes sense as long as $r \in
    O(x^{\frac12-\epsilon})$ for some $\epsilon>0$.  If additionally
    $~ln r \in o(~ln x)$ then $\widetilde{a}(x) \in 1+o(1)$.
    \ifshortversion\expandafter\qed\fi
  \end{theorem}
  You may want to sum this up as $\RSADM(x) \in (1+o(1)) \frac{4
    x}{~ln^{2} x} \left(~ln r - \frac{~ln r}{r} \right)$.  However,
  you then forego the option of actually calculating
  $\widetilde{a}(x)$.
  \begin{fullversion}
    \begin{proof}
      Consider $x$ large enough such that all sum boundaries are
      beyond $2657$, i.e.\ $\frac{\sqrt{x}}{r} \geq 2657$.  By
      definition $\RSADMSET$ is a notion of tolerance~$r$.  Further it
      is $\left[ c_{1}, c_{2} \right]$-balanced with $c_{1} = ~log_{x}
      \left( \frac{\sqrt{x}}{r} \right)= \frac{1}{2}- \frac{~ln r}{~ln
        x}$ and $c_{2} = ~log_{x} \left( \sqrt{rx} \right)=
      \frac{1}{2}+\frac{~ln r}{2 ~ln x}$.  As depicted next to
      \ref{def:RSA1}, we treat the upper half of the notion as the
      union of those two notions matching the two double sums in
      \ref{eq:RSA1SUM}, which both inherit being
      $[c_{1},c_{2}]$-balanced of tolerance~$r$.  Considering the
      inner bounds ${\frac{x}{r p}}$~to~${r p}$ and
      ${p}$~to~${\frac{x}{p}}$, respectively, as a function of the
      outer variable~$p$, we observe that the lower and upper bound in
      each case have opposite monotonicity behavior and thus by
      \ref{res:criterionMonotone} each part is a monotone notion.  We
      can thus apply \ref{res:psalm23}.
      Under the restriction $~ln r \in o(~ln x)$ we have $c_{1}, c_{2} \in
      \frac{1}{2}+o(1)$, which implies that $\frac{1}{c_{i}^{2}} \in 4
      \left(1+o(1)\right)$ for both $i \in \{ 1,2\}$.  Computing the
      area of the two parts yields
      $$
      \frint{\frint{1}{q}{\frac{x}{r p}}{r
          p}}{p}{\frac{\sqrt{x}}{r}}{\sqrt{\frac{x}{r}}} = \frac{1}{2}
      \cdot x \left(1 - \frac{~ln r}{r} - \frac{1}{r}\right)
      $$
      and
      $$
      \frint{\frint{1}{q}{p}{\frac{x}{p}}}{p}{\sqrt{\frac{x}{r}}}{\sqrt{x}}
      = \frac{1}{2} \cdot x \left(~ln r - 1 + \frac{1}{r} \right).
      $$
      For the error term we obtain $\mathcal{O}(x^{\frac{3}{4}}
      r^{\frac{1}{2}})$ noting that the number $\pi\left( \sqrt{x}
      \right)$ of prime squares up to $x$ is at most $\sqrt{x}$.
    \end{proof}
    Actually, we can even prove that the error term is in $\mathcal{O}
    \left( x^{\frac34} r^{\frac14} \right)$. We lost this
    \begin{fullversion}[due to the general form chosen for
      \ref{res:psalm23}.]%
      in the last steps of the proof of \ref{res:psalm23} when we
      replaced $C_{2}(p,x) = r p$ by $x/p$.
    \end{fullversion}
  \end{fullversion}
\end{fullversion}

\subsection{A fixed bound notion}\label{sec:fixedBound}

\begin{fullversion}
  A second possible definition for RSA integers can be stated as
  follows:%
\end{fullversion}
We consider the number of integers smaller than a real positive
bound~$x$ that have exactly two prime factors $p$ and $q$, both lying
in a fixed interval $]B, C]$, in formula:
\begin{equation*}
  \Qclean{2}(x)
  :=
  \#\Set{
    n \in \N;
    \begin{array}{c}
      \exists p, q \in \P{B,C}\colon\\
      n=p q \Land n \leq x
    \end{array}
  }.
\end{equation*}
To avoid problems with rare prime squares, which are also not
interesting when talking about RSA integers, we instead count
\begin{fullversion}
  \hanghere{\psset{unit=1.5mm} \drawnotion[]{8}{3}{\RSAFIXdraw{0.45}}}
\end{fullversion}
\begin{equation*}
  \label{def:Q}
  \Q{2}(x)
  :=
  \#\Set{
    (p,q) \in \left(\P{B,C}\right)^{2}; p q \leq x
  }.
\end{equation*}
Such functions are treated in \ifanonymous Anonymous (2010) \else
\cite{loenus10c} \fi.\iffullversion\\\fi In the context of RSA
integers we consider the notion
\begin{fullversion}[\refstepcounter{equation}]
  \begin{equation}
    \label{def:RSA3}
    \RSAFIXSET[r,\sigma]
    := \left\langle
      \Set{
        (y,z) \in \R_{>1}^{2};
        \sqrt{\frac{x}{r}} < y,z \leq \sqrt{r^{\sigma} x} \Land y z \leq x
      } \right\rangle_{x \in \R_{>1}}
  \end{equation}
\end{fullversion}%
\begin{shortversion}%
  \begin{equation*}
    \label{def:RSA3}
    \RSAFIXSET[r,\sigma]
    := \left\langle
      \Set{
        (y,z) \in \R_{>1}^{2};
        \sqrt{\frac{x}{r}} < y,z \leq \sqrt{r^{\sigma} x} \Land y z \leq x
      } \right\rangle_{x \in \R_{>1}}
  \end{equation*}
\end{shortversion}%
with $\sigma \in [0,1]$. The parameter $\sigma$ describes the
(relative) distance of the restriction $yz \leq x$ to the center of
the rectangle in which $y$ and $z$ are allowed.
\begin{fullversion}[\refstepcounter{equation}]%
  We split the corresponding counting function into two double sums:
  \begin{equation}\label{eq:RSA3SUM}
    \RSAFIX[r,\sigma](x) = 
    \begin{aligned}[t]
      &\primesum{ \primesum{1}{q}{\sqrt{\frac{x}{r}}}{\sqrt{r^{\sigma}
            x}} }{p}{\sqrt{\frac{x}{r}}}{\sqrt{\frac{x}{r^{\sigma}}}}
      \\
      &+ \primesum{ \primesum{1}{q}{\sqrt{\frac{x}{r}}}{\frac{x}{p}}
      }{p}{\sqrt{\frac{x}{r^{\sigma}}}}{\sqrt{r^{\sigma} x}}.
    \end{aligned}
  \end{equation}
\end{fullversion}%
The next theorem follows directly from \ifanonymous Anonymous (2010)
\else \cite{loenus10c} \fi but 
we can also derive it from \ref{res:psalm23}%
\begin{fullversion}[:]
  similar to \ref{res:order-rsa1}.
\end{fullversion}
\begin{theorem}\label{res:order-rsa3}
  We have under the Riemann hypothesis
  $$
  \RSAFIX[r,\sigma](x) \in \widetilde{a}(x) \frac{4x}{~ln^{2} x}
  \left( \sigma ~ln r + 1 - \frac{2}{r^{\frac{1-\sigma}{2}}} +
    \frac{1}{r} \right) + \bigO{ x^{\frac{3}{4}} r^{\frac{1}{4}} }
  $$
  with $\widetilde{a}(x) \in \left[ \left(1 - \frac{\sigma ~ln r}{~ln
        x + \sigma ~ln r}\right)^2, \left(1 + \frac{~ln r}{~ln x - ~ln
        r}\right)^2\right]$.  If additionally $~ln r \in o(~ln x)$
  then $\widetilde{a} \in 1+o(1)$.
  \ifshortversion\expandafter\qed\fi
\end{theorem}
\begin{fullversion}
  \begin{proof}
    Let $x$ be such that all sum boundaries are beyond $2657$. By
    definition $\RSAFIXSET$ is a notion of tolerance~$r$.  Further for
    all $\sigma \in [0,1]$ it is clearly $\left[ c_{1}, c_{2}
    \right]$-balanced with $c_{1} = ~log_x \sqrt{\frac{x}{r}} =
    \frac{1}{2}- \frac{~ln \sqrt{r}}{~ln x}$ and $c_{2} = ~log_x
    \sqrt{r^{\sigma} x} = \frac{1}{2}+\frac{\sigma ~ln r}{2 ~ln x}$.
    As depicted next to \ref{def:RSA3}, we treat the notion as the
    union of two notions corresponding to the two double sums in
    \ref{eq:RSA3SUM}, which are both $[c_{1},c_{2}]$-balanced of
    tolerance~$r$.

    Consider the inner bounds
    $\sqrt{\frac{x}{r}}$~to~$\sqrt{r^{\sigma} x}$ and
    $\sqrt{\frac{x}{r}}$~to~$\frac{x}{p}$ respectively, as a function
    of the outer variable~$p$ (while $\sigma$ is fixed): We observe
    that the lower and upper bound in the first case are constant and
    in the second case consist of a constant lower bound and an
     weakly decreasing upper bound.  Thus by \ref{res:criterionMonotone} each
    part is a monotone notion and we can apply \ref{res:psalm23}.
    
    As in the proof of \ref{res:order-rsa1}, we have under the additional restriction $\ln r \in o(~ln(x))$ that $\frac{1}{c_{i}^{2}} \in 4 \left(1+o(1)\right)$ for both $i \in \{
    1,2\}$.  Computing the area of the two parts yields
    $$
    \frint{\frint{1}{q}{\sqrt{\frac{x}{r}}}{\frac{x}{p}}
    }{p}{\sqrt{\frac{x}{r^{\sigma}}}}{\sqrt{r^{\sigma} x}} = x \left(
      \sigma ~ln r + \frac{1}{r^{(1+\sigma)/2}} -
      \frac{1}{r^{(1-\sigma)/2}} \right)
    $$
    and
    $$
    \frint{ \frint{1}{q}{\sqrt{\frac{x}{r}}}{\sqrt{r^{\sigma}
          x}}}{p}{\sqrt{\frac{x}{r}}}{\sqrt{\frac{x}{r^{\sigma}}}} = x
    \left(1 - \frac{1}{r^{(1-\sigma)/2}} - \frac{1}{r^{(1+\sigma)/2}}
      + \frac{1}{r} \right)
    $$
    For the error term we obtain $\bigO{x^{\frac{3}{4}}
      r^{\frac{1}{4}}}$.
  \end{proof}
\end{fullversion}

\subsection{An algorithmically inspired notion}

A \ifshortversion second\else third \fi
option to define RSA integers is the following notion: Assume you wish
to generate an RSA integer between $\frac{x}{r}$ and $x$, which has
two prime factors of roughly equal size. Then algorithmically we might
first generate the prime $p$ and afterward select the prime $q$ such
that the product is in the correct interval. As we will see later,
this procedure does ---~however~--- not produce every number with the
same probability, see \ref{sec:genProperly}.  Formally, we consider
the notion

\begin{fullversion}[\refstepcounter{equation}\refstepcounter{equation}\refstepcounter{equation}\refstepcounter{equation}]
  \begin{fullversion}
    \hanghere[0\baselineskip]{\psset{unit=1.5mm}\drawnotion[]{8}{3}{\RSAALGdraw}}
  \end{fullversion}%
  \begin{equation}
    \label{def:RSA2}
    \RSAALGSET
    := \left\langle 
      \Set{ (y,z) \in \R_{>1}^{2};
        \begin{array}{c}
          \frac{\sqrt{x}}{r} < y \leq \sqrt{x},
          \\
          \frac{x}{r y} < z \leq \frac{x}{y},
          \\
          \frac{x}{r} < y z \leq x
        \end{array}
      } \right\rangle_{x \in \R_{>1}.}
    \qquad\qquad
  \end{equation}%
  \begin{fullversion}[\refstepcounter{equation}\refstepcounter{equation}\refstepcounter{equation}]%
    We proceed with this notion similar to the previous one. By
    observing
    \begin{equation}\label{eq:RSA2SUM}
      \RSAALG(x) =
      \begin{aligned}[t]
        &\primesum{ \primesum{1}{q}{\sqrt{x}}{\frac{x}{p}}
        }{p}{\frac{\sqrt{x}}{r}}{\sqrt{x}}\\
        &+ \primesum{ \primesum{1}{q}{\frac{x}{r p}}{\sqrt{x}}
        }{p}{\frac{\sqrt{x}}{r}}{\sqrt{x}},
      \end{aligned}
    \end{equation}
    and again applying \ref{res:psalm23} and
    \ref{res:criterionMonotone} we obtain
  \end{fullversion}
  \begin{shortversion}
    We obtain for this notion using \ref{res:psalm23}
  \end{shortversion}
  \begin{theorem}
    \label{res:order-rsa2}
    We have under the Riemann hypothesis
    $$
    \RSAALG(x) \in \widetilde{a}(x) \frac{4 x}{~ln^{2} x} \left(~ln r
      - \frac{~ln r}{r} \right) + \mathcal{O}\left( x^{\frac{3}{4}}
      r^{\frac{1}{2}} \right)
    $$
    with $\widetilde{a}(x) \in \left[ \left(1 - \frac{2 ~ln r}{~ln x +
          2 ~ln r}\right)^2, \left(1 + \frac{2 ~ln r}{~ln x - 2 ~ln
          r}\right)^2\right]$. If additionally  $~ln r \in o(~ln x)$
    then $\widetilde{a} \in 1+o(1)$.
    \ifshortversion\expandafter\qed\fi
  \end{theorem}
  \begin{fullversion}
    \begin{proof}
      Again let $x$ be such that all sum boundaries are beyond
      $2657$. By definition $\RSAALGSET$ is a notion of tolerance~$r$.
      Further it is clearly $\left[c_{1}, c_{2} \right]$-balanced with
      $c_{1} = ~log_x \frac{\sqrt{x}}{r} = \frac{1}{2}- \frac{~ln
        r}{~ln x}$ and $c_{2} = ~log_x r \sqrt{x} = \frac{1}{2} +
      \frac{~ln r}{~ln x}$.  As depicted next to \ref{def:RSA2}, we
      treat the notion as the union of two notions corresponding to
      the two double sums in \ref{eq:RSA2SUM}, which are both
      $[c_{1},c_{2}]$-balanced of tolerance~$r$.

      If we consider the inner bounds $\sqrt{x}$~to~$\frac{x}{p}$ and
      $\frac{x}{r p}$~to~$\sqrt{x}$, respectively, as a function of
      the outer variable~$p$, we observe that in both cases one of
      them is constant and the other decreasing. Furthermore by
      \ref{res:criterionMonotone} each part is a monotone notion. We
      can thus apply \ref{res:psalm23}.

As for the previous notions we have under the additional restriction $\ln r \in o(~ln(x))$ that $\frac{1}{c_{i}^{2}} \in 4 \left(1+o(1)\right)$ for both $i \in \{ 1,2\}$.  
      Computing the
      area of the two parts yields
      $$
      \frint{\frint{1}{q}{\sqrt{x}}{\frac{x}{p}}}{p}{\frac{\sqrt{x}}{r}}{\sqrt{x}}
      = x \left(1 - \frac{~ln r}{r} - \frac{1}{r} \right)
      $$
      and
      $$
      \frint{\frint{1}{q}{\frac{x}{r
            p}}{\sqrt{x}}}{p}{\frac{\sqrt{x}}{r}}{\sqrt{x}} = x
      \left(~ln r - 1 + \frac{1}{r} \right).
      $$
      For the error term we obtain $\bigO{x^{\frac{3}{4}}
        r^{\frac{1}{2}}}$.
    \end{proof}
    Note that we also could have employed
    \ref{res:criterionMonotoneSplit}, but in this particular case we
    decided to use another split of the notion.
  \end{fullversion}

  The IEEE standard P1363 suggest a slight variant, both generalize to
  \begin{equation}
    \RSAALGVARSET[r,\sigma](x)
    := \left\langle 
      \Set{ (y,z) \in \R_{>1}^{2};
        \begin{array}{c}
          r^{\sigma-1} \sqrt{x} < y \leq r^{\sigma} \sqrt{x},\\
          \frac{x}{r y} < z \leq \frac{x}{y},
          \\
          \frac{x}{r} < y z \leq x
        \end{array}
      } \right\rangle_{x \in \R_{>1}},
  \end{equation}%
  with $\sigma \in [0,1]$.
\end{fullversion}
\begin{shortversion}
  \begin{equation*}
    \RSAALGVARSET[r, \sigma](x)
    := \left\langle 
      \Set{ (y,z) \in \R_{>1}^{2};
        \begin{array}{c}
          r^{\sigma-1} \sqrt{x} < y \leq r^{\sigma} \sqrt{x}
          \Land \frac{x}{r y} < z \leq \frac{x}{y}
          \\
          \frac{x}{r} < y z \leq x
        \end{array}
      } \right\rangle_{x \in \R_{>1}},
  \end{equation*}%
  with $\sigma \in [0,1]$. The parameter $\sigma$ describes here the
  (relative) position of the defining area of the notion with respect
  to the diagonal.  Write $\sigma' := ~max(\sigma, 1-\sigma)$.
\end{shortversion}
\begin{fullversion}
  Now, our notion above is $\RSAALGVARSET[r,0]$, and the IEEE variant
  is $\RSAALGVARSET[r,\frac12]$.  By similar reasoning as above we
  obtain
\end{fullversion}
\begin{shortversion}
  Similar to the theorem above we obtain
\end{shortversion}
\begin{theorem}
  \label{res:order-rsa2b}
  We have under the Riemann hypothesis
  $$
  \RSAALGVAR[r, \sigma](x) \in%
  \widetilde{a}(x) \frac{4x}{~ln^{2}x} \left( ~ln r - \frac{~ln r}{r}
  \right) + \bigO{x^{\frac{3}{4}} r^{\frac{1}{2}}},
  $$
  with $\widetilde{a}(x) \in \left[ \left(1 - \frac{2 \sigma' ~ln
        r}{~ln x + 2 \sigma' ~ln r}\right)^2, \left(1 + \frac{2
        (1+\sigma) ~ln r}{~ln x - 2 (1+\sigma) ~ln
        r}\right)^2\right]$, where $\sigma' = ~max(\sigma, 1-\sigma)$.
  If additionally $~ln r \in o(~ln x)$ then $\widetilde{a}(x) \in
  1+o(1)$ \qed
\end{theorem} 

\iffullversion
\subsection{Summary}
\fi As we see, \begin{fullversion}[both ]all\end{fullversion} notions%
\begin{fullversion}[~]%
  , summarized in \ref{fig:RSAThreeNotions},%
\end{fullversion}
open a slightly different view.  However the outcome is not that
different, at least the numbers of described RSA integers are quite
close to each other, see \ref{sec:CompareNotions}.
\begin{shortversion} The proof that this is the case for \emph{all}
  reasonable notions can be found in the extended
  version\ifanonymous\else\ \cite{loenus11a}\fi.
\end{shortversion}


Current standards and implementations of various crypto packages
mostly use the notions $\RSAFIXSET[4,0]$, $\RSAFIXSET[4,1]$,
$\RSAFIXSET[2,0]$ or $\RSAALGVARSET[2,1/2]$.  For details see
\ref{sec:concImplem}.

\begin{fullversion}
  \section{Arbitrary notions}
  \label{sec:CompareNotions}

  The preceding examinations show that the order of the analyzed
  functions differ by a factor that only depends on the notion
  parameters, i.e.\ on~$r$ and $\sigma$,
  \begin{fullversion}[see \ref{res:order-rsa1},
    \bare\ref{res:order-rsa3} and \bare\ref{res:order-rsa2}. ]
    summarizing:
    \begin{theorem*}
      \label{res:RSAapproxSame}
      Assuming $~ln r \in o(~ln x)$ and $r>1$ and $\sigma \in [0,1]$
      we have
      \begin{enumerate}
      \item $\RSADM(x) \in (1+o(1)) \frac{4 x}{~ln^{2} x} \left(~ln r
          - \frac{~ln r}{r} \right),$
      \item $\RSAFIX[r,\sigma](x) \in (1+o(1)) \frac{4x}{~ln^{2} x}
        \left( \sigma ~ln r + 1 - \frac{2}{r^{\frac{1-\sigma}{2}}} +
          \frac{1}{r} \right)$,
      \item $\RSAALG(x) \in (1+o(1)) \frac{4 x}{~ln^{2} x} \left(~ln r
          - \frac{~ln r}{r} \right).$ \qed
      \end{enumerate}
    \end{theorem*}
  \end{fullversion}
  It is obvious that the three considered notions with many parameter
  choices cover about the same number of integers.

  To obtain a much more general result, we consider the following
  maximal notion
  \hanghere[-0.3\baselineskip]{\psset{unit=1.5mm}\drawnotion[]{8}{3}{\RSAMAXdraw}}
  \begin{equation}
    \label{def:RSAMAX}
    \RSAMAXSET
    := \left\langle 
      \Set{ (y,z) \in \R_{>1}^{2};
        \begin{array}{c}
          x^{c_{1}} < y \leq x^{1-c_{1}},
          \\
          x^{c_{1}} < z \leq x^{1-c_{1}},
          \\
          \frac{x}{r} < y z \leq x
        \end{array}
      } \right\rangle_{x \in \R_{>1}.}
    \qquad\qquad
  \end{equation}%
  All of the notions discussed in \ref{sec:NotionsInParticular} are
  subsets of this notion. Using the same techniques as above, we
  obtain:
  \begin{theorem}\label{res:order-rsamax}
    For $~ln r \in o(~ln x)$ we have under the Riemann hypothesis
    \begin{enumerate}
    \item\label{res:order-rsamax-i} For $c_1 \leq \frac{1}{2} - \frac{~ln r}{2~ln x}$ and for some $l$ and large $x$ additionally $c_1 > \frac12 - ~ln^l x ~ln r$, we have that 
      \begin{small}
      $$
      \RSAMAX(x) \in \widetilde{a}(x) \frac{4x}{~ln^{2} x} \left( (1 -
        2 c_{1}) \left(1-\frac{1}{r}\right) ~ln x - 1 + \frac{~ln r +
          1}{r} \right) + \mathcal{O}\left(c^{-1}
        x^{1-\frac{c_{1}}{2}}  ~ln^{l+1}\right),
      $$
      \end{small}
    \item\label{res:order-rsamax-ii} when $c_1 > \frac{1}{2} - \frac{~ln r}{2~ln x}$, we obtain the fixed bound notion
      $$
      \RSAMAX(x) \in \widetilde{a}(x) \frac{4x}{~ln^{2} x} \left( (1 -
        2 c_{1}) ~ln x + \frac{1}{x^{1-2 c_{1}}} - 1 \right) +
      \mathcal{O}\left( c_1^{-1} \cdot x^{1-\frac{c_{1}}{2}}
        x\right)\hspace{2cm}.
      $$
      This is independent of $r$.
    \end{enumerate}
    In both cases $\widetilde{a}(x) \in \left[ \frac{1}{4(1-c_{1})^2},
      \frac{1}{4c_{1}^2} \right]$.  In
    particular for $c_1 \in \frac{1}{2}+o(1)$ we have
    $\widetilde{a}(x) \in 1+o(1)$. \ifshortversion\qed\fi
  \end{theorem}
  \begin{fullversion}
 Case \ref{res:order-rsamax-i} considers the case where the notion $\RSAMAXSET$ looks like a thin band. The other alternative \ref{res:order-rsamax-ii} treats the case where the notion is actually a triangle, namely the notion $\RSAFIXSET[x^{1-2c_1},1]$. In the former case we have to make sure that the band is not too long so that we may apply  \ref{res:psalm23} for not too many pieces. As noted after \ref{def:notion-special}, the first case could still be somewhat extended.
    \begin{proof}
      As usual let $x$ be such that all sum boundaries are beyond
      $2657$. By definition $\RSAMAXSET$ is a notion of tolerance~$r$.
      Further it is clearly $\left[c_{1}, 1-c_{1}
      \right]$-balanced. For $c_1 > \frac{1}{2} - \frac{~ln r}{2~ln
        x}$ the result follows directly from \ref{res:order-rsa3},
      since $\RSAMAXSET$ is simply the fixed bound notion
      $\RSAFIXSET[x^{1-2c_1},1]$.
      
      For $c_1 \leq \frac{1}{2} - \frac{~ln r}{2~ln x}$ we treat the
      notion as the sum of several monotone, $[c_1, 1-c_1]$-balanced
      notions of tolerance $r$ by triangulating the maximal notion as
      indicated in the picture next to \ref{def:RSAMAX}. The number
      $m$ of necessary cuts is $(1-2c_1) ~\frac{~ln x}{~ln r}$ which
      is in $\bigO{~ln^{l+1} x}$ by assumption. This gives by \ref{res:psalm23} the
      claim.
    \end{proof}
  \end{fullversion}
  We obtain
  \begin{theorem}
    \label{res:RSAEquivalent}
    Let $c_{1}, c_{2} \in \frac{1}{2}+o(1)$, $r>1$ with $~ln r \in
    \bigOmega{\frac{1-2c_1}{~ln^l x}} \cap o(~ln x)$ be possibly
    $x$-dependent values, and $a \in \left]0,1\right[$ constant.
    Consider a
    \begin{fullversion}%
      piecewise
    \end{fullversion}%
    monotone notion $\RSASET$ of RSA integers with tolerance~$r$ such
    that for large $x \in \R_{>1}$ we have $~area \RSASET_{x} \geq a
    x$.  Then
    \begin{align*}
      \RSA(x) &= \frac{4x}{~ln^{2}x} \cdot \widetilde{a}(x)
    \end{align*}
    where $\widetilde{a}(x) \in o(~ln x)$ and $\widetilde{a}(x) \geq a
    - \epsilon(x)$ for some $\epsilon(x) \in o(1)$.


    In particular, the prime pair counts of two such notions can
    differ by at most a factor of order $o(~ln
    x)$. \ifshortversion\qed\fi
  \end{theorem}
  \begin{shortversion}
    The statement of the theorem follows basically by an application
    of \ref{res:psalm23} and \ref{res:order-rsamax}.  Clearly the
    theorem remains true if the notion is a union of, say, two or
    three monotone notions as, for example, $\RSADMSET$ or
    $\RSAALGSET$.
  \end{shortversion}
  \begin{fullversion}
    \begin{proof}
      Let $\RSASET$ be as specified.  Assume $x$ to be large enough to
      grant that $~area \RSASET_{x} \geq a x$ and $x^{c_{1}} > 2657$.
      Without loss of generality we assume $c_{1}+c_{2}\leq 1$.
      Otherwise we replace $c_{2} = 1-c_{1}$.  Denote $c :=
      ~max(2c_{2}-1,1-2c_{1})$, this now is always in $[0,1]$.
      By \ref{res:psalm23} we obtain
      $$
      \RSA(x) \geq a \frac{4x}{~ln^{2} x} - \widehat{a}(x),\qquad
      \widehat{a}(x) \in \mathcal{O}(x^{\frac{3+c}{4}}).
      $$
      To provide an upper bound, we consider the $[c_{1},
      1-c_{1}]$-balanced maximal notion \ref{def:RSAMAX}.
      As mentioned above we have for all $x \in \R_{>1}$ that
      $\RSASET_{x} \subseteq \RSAMAXSET_{x}$, and so $\RSA(x) \leq
      \RSAMAX(x)$.  Note that $c_{1} \leq \frac{1}{2}$, as otherwise
      $\mathcal{A}_{x}$ would be empty rather than having area at
      least $a x$.  By assumption we have $c_{1} \in \frac{1}{2}+o(1)$
      and thus $0\leq 1-2c_{1} \in o(1)$. Now the claim follows from
      \ref{res:order-rsamax} and the assumption $~ln r \in o(~ln x)$.
    \end{proof}
  \end{fullversion}

  \begin{fullversion}
    In the following we will analyze the relation between the proposed
    notions in more detail.  Namely, we carefully check how each of
    the notions can be enclosed in terms of the others. Clearly the
    fixed bound notions $\RSAFIXSET[r, \sigma]$
    \begin{fullversion}[grow with
      increasing~$\sigma$.\refstepcounter{equation} ]%
      enclose each other:
      \begin{lemma}\label{res:compareCountRSA3}
        For $r \in \R_{>1}$, $x \in \R_{>1}$ and $\sigma, \sigma' \in
        [0,1]$ with $\sigma \leq \sigma'$ we have
        $$
        \RSAFIX[\sqrt{r}, 1](x/\sqrt{r}) \leq \RSAFIX[r, 0](x) \leq
        \RSAFIX[r, \sigma](x) \leq \RSAFIX[r, \sigma'](x) \leq
        \RSAFIX[r,1](x)
        $$
      \end{lemma}
      \begin{proof}
        For the first inequality simply observe that $x/\sqrt{r} \leq
        x$.  The remaining inequalities follow from the fact that
        $\sqrt{r^{\sigma} x} \leq \sqrt{r^{\sigma'} x}$ whenever
        $\sigma \leq \sigma'$.
      \end{proof}
    \end{fullversion}
    We can also enclose different notions by each other:
    \begin{lemma}\label{res:compareCountRSAAll}
      For $r \in \R_{>1}$ and $x \in \R_{>1}$ we have
      $$
      \frac{1}{2} \RSAFIX[r,1](x) \leq \frac{1}{2}\RSADM(x) \leq
      \RSAALG(x) \leq \RSAFIX[r^{2},1](x)
      $$
    \end{lemma}
    \begin{fullversion}
      \begin{proof}
        We prove every inequality separately.
        \begin{figure}
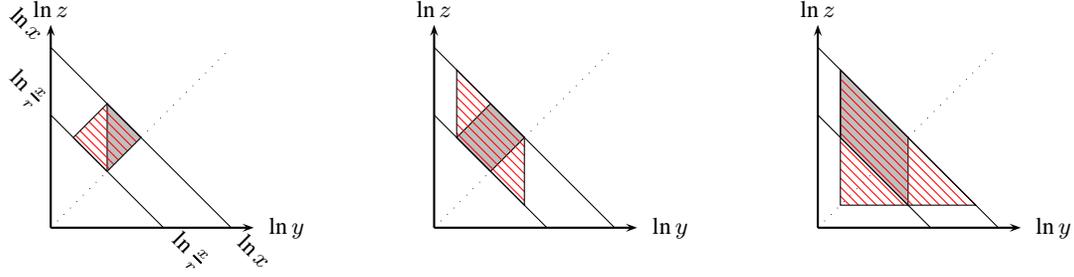

          \centering%
          \newpsstyle{notionsplit}{linestyle=none}%
          \drawnotion{8}{3}{%
            \begin{psclip}{\psline[linestyle=none]%
                (-\mylnx,-\mylnx)(\mylnx,\mylnx)(\mylnx,-\mylnx)}
              \RSAFIXdraw{1}%
              \newpsstyle{notionarea}{style=notionarea2} \RSADMdraw
            \end{psclip}
            \psline[style=notionarea](0,0)(1,1)%
          }%
          \hfill%
          \drawnotion[]{8}{3}{%
            \begin{psclip}{\psline[linestyle=none]%
                (-\mylnx,-\mylnx)(\mylnx,\mylnx)(\mylnx,-\mylnx)}
              \RSADMdraw
            \end{psclip}
            \psline[style=notionarea](0,0)(1,1)%
            \newpsstyle{notionarea}{style=notionarea2} \RSAALGdraw%
          }%
          \hfill%
          \drawnotion[]{8}{3}{%
            \begin{psclip}{\psline[linestyle=none]%
                (-\mylnx,-\mylnx)(\mylnx,\mylnx)(\mylnx,-\mylnx)}
            \end{psclip}
            \RSAALGdraw%
            {%
              \newpsstyle{notionarea}{style=notionarea2}
              \psset{unit=2\psunit}
              \RSAFIXdraw{1}}%
          }%
          \caption{Enclosing notions of RSA integers using others.}
          \label{fig:RSAcompareAll}
        \end{figure}
        For an easier understanding of the proof a look at
        \ref{fig:RSAcompareAll} is advised:
        \begin{description}
        \item[$\frac{1}{2}{\RSAFIX[r,1](x)} \leq
          \frac{1}{2}{\RSADM(x)}$: ] Consider the double sum
          \ref{eq:RSA3SUM}
          $$
          \frac{1}{2} \RSAFIX[r,1](x) = \frac{1}{2} \primesum{
            \primesum{1}{q}{\sqrt{\frac{x}{r}}}{\frac{x}{p}}} {p}
          {\sqrt{\frac{x}{r}}} {\sqrt{r x}} = \primesum{
            \primesum{1}{q}{p}{\frac{x}{p}}} {p} {\sqrt{\frac{x}{r}}}
          {\sqrt{x}}
          $$
          due to the restriction $p < q$. This is exactly the second
          summand in \ref{eq:RSA1SUM}.
        \item[$\frac{1}{2}{\RSADM(x)} \leq {\RSAALG(x)}$: ] Consider
          again the double sum \ref{eq:RSA1SUM}. We expand the
          summation area for $q$ (thus increasing the number of
          primepairs we count) in order to obtain the sum
          \ref{eq:RSA2SUM} for the algorithmic notion: For the first
          summand we obtain from $p \leq \sqrt{\frac{x}{r}}$ that $rp
          \leq \frac{x}{p}$ and for the second summand from the same
          argument that $\frac{x}{rp} \leq p$. The third summand
          disappears while doing this, since the squares (which are
          counted by the third summand) are now counted by the second
          summand. Thus we can bound the whole sum from above by
          changing the summation area for $q$ in this way.
        \item[${\RSAALG(x)} \leq {\RSAFIX[r^{2},1](x)}$: ] We proceed
          as in the previous step, by replacing in the sum
          \ref{eq:RSA2SUM} the summation area for $q$: Since $p \leq
          \sqrt{x}$, we obtain $\frac{x}{rp} \geq
          \frac{\sqrt{x}}{r}$. Now since $\sqrt{x} \leq r \sqrt{x}$
          the claim follows.\qed
        \end{description}
      \end{proof}
    \end{fullversion}
    We actually can enclose the \citeauthor{decmor08} notion even
    tighter by the fixed bound
    notion:
    \begin{lemma}
      For $r \in \R_{>1}$ and $x \in \R_{>1}$ we have
      $$
      \RSAFIX[r,1](x) \leq \RSADM(x) \leq
      \RSAFIX[r^{2},\frac{1}{2}](x).
      $$
    \end{lemma}
    \begin{fullversion}
      \begin{proof}
        Assume $\sqrt{\frac{x}{r}} < p < q \leq \sqrt{r x}$ and $p q
        \leq x$. Then $\frac{x}{r} < p q \leq x$ and $q \leq r p$. If
        on the other hand $\frac{x}{r} < p q \leq x$ and $p < q < r
        p$, then $\frac{x}{r^{2}} < \frac{1}{r} p q < p^{2} < q^{2} <
        r p q \leq r x$ and the claim follows.
      \end{proof}
      All the inclusion described above are  compatible to the result from \ref{res:RSAEquivalent}. However, many of the explicit inclusions are much tighter.
    \end{fullversion}

    \begin{shortversion}
      \noindent
      Please note that proofs for the results in this section are
      postponed to the extended
      version\ifanonymous\else\cite{loenus11a}\fi.
    \end{shortversion}
  \end{fullversion}
\end{fullversion}


\section{Generating RSA integers}\label{sec:genProperly}
\def\leftflag#1{}

In this section we analyze how to generate RSA integers properly.  It
completes the picture and we found several implementations overlooking
this kind of arguments.

\noindent We wish that all the algorithms generate integers with the
following properties:
\begin{itemize}
\item If we fix $x$ we should with at least overwhelming probability
  generate integers that are a product of a prime pair in
  $\RSASET_{x}$.
\item These integers (not the pairs) should be selected roughly
  uniformly at random.
\item The algorithm should be efficient.  In particular, it should
  need only few primality tests.
\end{itemize}
\begin{fullversion}
  For the first point note that we usually use probabilistic primality
  tests with a very low error probability, for example \cite{mil76},
  \cite{rab80}, \cite{solstr77}, or \cite{art66a}
  . 
  Deterministic primality tests are also available but at present for
  these purposes by far too slow.\nocite{agrkay04}
\end{fullversion}%

\subsection{Rejection sampling}
Assume that $\RSASET$ is a $[c_{1},c_{2}]$-balanced notion of RSA
integers with tolerance~$r$.  The easiest approach for generating a
pair from $\RSASET$ is based on \citeauthor{neu51}'s rejection
sampling method.
\begin{fullversion}[\refstepcounter{equation}]%
  For this the following definition comes in handy:
  \hanghere[.4\baselineskip]{%
    \psset{unit=1mm,labelsep=1mm}
    \begin{pspicture}(-2,0)(12,16) \pscustom[style=notionarea]{%
        \pscurve(0,4)(2,5)(6,3)(8,3.5)%
        \psline(8,3.5)(8,9.5)%
        \pscurve[liftpen=1](8,9.5)(6,9.2)(2,10.7)(0,10)%
        \psline(0,10)(0,3.93)%
      }%
      \psline[linewidth=1pt, linestyle=dotted](6,1)(6,12)
    \end{pspicture}
  }
  \begin{definition}[Banner]
    \rightskip=18mm%
    A \emph{banner} is a graph-bounded notion of RSA integers such
    that for all $x \in \R_{>1}$ and for every prime $p \in
    [B_{1}(x),C_{1}(x)]$ the number $f_{x}(p)$
    of primes in the interval $[B_{2}(p,x),C_{2}(p,x)]$ is almost
    independent of~$p$ in the following sense: $ \frac {~max\Set{
        f_{x}(p); p \in [B_{1}(x),C_{1}(x)] \cap \PR}} {~min\Set{
        f_{x}(p); p \in [B_{1}(x),C_{1}(x)] \cap \PR}} \in
    1+o\treatparameter(1).  $
  \end{definition}

  For example, a rectangular notion, where $B_{2}(p,x)$ and
  $C_{2}(p,x)$ do not depend on $p$, is a banner.  Now given any
  notion~$\RSASET$ of RSA integers we select a banner $\mathcal{B}$ of
  (almost) minimal area enclosing $\mathcal{A}_{x}$.  Note that there
  may be many choices for $\mathcal{B}$.  We can easily generate
  elements in $\mathcal{B}_{x} \cap \N^{2}$: Select first an
  appropriate $y \in [ B_{1}(x), C_{1}(x) ] \cap \N$, second an
  appropriate $z \in [ B_{2}(p,x), C_{2}(p,x) ] \cap \N$.  By the
  banner property this chooses $(y,z)$ almost uniformly.
  \begin{skipme}
    \begin{definition}[Banner]
      Let $I_{1}, I_{2} \subset \R$ be two intervals and
      $\shortmap[f]{I_{1}}{\R}$ be a map. \par\pushright{\vbox
        to0pt{\hfill
          \input{\jobname-banner}\vss}}\vskip-\baselineskip\vskip0pt%
      \noindent
      A \emph{banner} is any subset $A \subset \R^{2}$ of the form
      $$\Set{(y,z) \in \R^{2}; y \in I_{1}, z-f(y) \in I_{2} }.$$
    \end{definition}
    Now given any notion~$\RSASET$ of RSA integers we select for every
    $x \in \R_{>1}$ one banner $\mathcal{B}_{x}$ of minimal area
    enclosing $\mathcal{A}_{x}$. Note that there may be many choices
    for $\mathcal{B}_{x}$. We can in a standard way generate elements
    uniformly at random in $\mathcal{B}_{x} \cap \N^{2}$: Select first
    an appropriate $y \in I_{1} \cap \N$, second an appropriate $z \in
    (f(y)+I_{2}) \cap \N$.  We assume here that the number of primes
    in $f(y)+I_{2}$ depends only weakly on $y$, for example, between
    $\frac{~length(I_{2})}{~ln \sqrt{x} - ~ln r}$ and
    $\frac{~length(I_{2})}{~ln \sqrt{x} + ~ln r}$.  Using
    \ref{pnt-schoenfeld} we can refine the notion of a banner to
    ensure that asymptotically.  Clearly, a constant $f$ fulfills that
    precisely.
  \end{skipme}
\end{fullversion}
\begin{shortversion}
  Let $\mathcal{B}_x := x^{[c_{1}, c_{2}]} \times x^{[c_{1}, c_{2}]}$.
  There may be better ways for choosing $\mathcal{B}_{x} \supseteq
  \mathcal{A}_{x}$, but we skip this here.
\end{shortversion}%
We obtain the following straightforward Las Vegas algorithm:
\begin{algorithm}{algoRSALasVegas}[Generating an RSA integer (Las
  Vegas version)]
\item A notion $\RSASET$, a bound $x \in \R_{>1}$.
\item An integer $n=p q$ with $(p,q) \in \RSASET_{x}$.
\item \begin{blockuntil}{$y$ prime and $z$ prime.}
  \item \begin{blockuntil}{$(y,z) \in \RSASET_{x}$.}
    \item Select $(y,z)$ at random from $\mathcal{B}_{x} \cap \N^{2}$%
      \begin{fullversion}[.]
        ~as just described.
      \end{fullversion}
    \end{blockuntil}
  \end{blockuntil}
\item $p \gets y$, $q \gets z$.
\item \RETURN $p q$.
\end{algorithm}
The expected repetition count of the inner loop is $\frac{\#\mathcal{B}(x)}{\RSA(x)}$ which is roughly
$\frac{~area(\mathcal{B}_{x})}{~area(\RSASET_{x})}$.  The expected
number of primality tests is about
$\frac{~area(\RSASET_{x})}{\RSA(x)}$.  \begin{fullversion}[This ]%
 By
  \ref{res:RSAEquivalent} this \end{fullversion}%
is for many notions in $\bigO{~ln^2 x}$.  We have seen implementations
(for example the one of \texttt{GnuPG}) where the inner and outer loop
have been exchanged.  This increases the number of primality tests by
the repetition count of the inner loop.
\begin{fullversion}%
  For $\RSAFIXSET[r,1]$ this is a factor of about
  \hanghere[\baselineskip]{\psset{unit=1.3mm}\drawnotion[]{8}{4}{
      \newpsstyle{notionsplit}{linestyle=none} \rput(-1,-1){%
        \psset{unit=2\psunit}%
        \newpsstyle{notionarea}{style=bannerarea} \RSAFIXdraw{0}%
      }%
      \RSAFIXdraw{1} }}
  $$
  \frac {\RSAFIX[r^{2},0](r x)} {\RSAFIX[r,1](x)} \sim%
  \frac{\frac{1}{r} -2 + r}{~ln r + \frac1r - 1} =%
  \frac{ (r-1)^{2} }{r (~ln r - 1) + 1},
  $$
  which for $r=2$ is equal to
  $2.\roundeddown{58}{86994495620898308053844319421}$ and even worse
  for larger $r$.
\end{fullversion}%
Also easily checkable additional conditions, like $~gcd((p-1)(q-1),
e)=1$, should be checked before the primality tests to improve the
efficiency.

\subsection{Inverse transform sampling}
Actually we would like to avoid generating out-of-bound pairs
completely.
\begin{fullversion}[\refstepcounter{equation} ]%
  Then a straightforward attempt to construct such an algorithm looks
  the following way:
  \begin{algorithm}{algoRSANonUniform}[Generating an RSA integer
    (non-uniform version)]
  \item A notion $\RSASET$, a bound $x \in \R_{>1}$.
  \item An integer $n=p q$ with $(p,q) \in \RSASET_{x}$.
  \item \begin{blockuntil}{$y$ prime.}
    \item Select $y$ uniformly at random from $\Set{y \in \R; \exists
        z\in\N\colon (y,z) \in \RSASET_{x}} \cap \N$.
    \end{blockuntil}
  \item $p \gets y$.
  \item \begin{blockuntil}{$z$ prime.}
    \item Select $z$ uniformly at random from $\Set{z \in \R; (p,z)
        \in \RSASET_{x}} \cap \N$.
    \end{blockuntil}
  \item $q \gets z$.
  \item \RETURN $p q$.
  \end{algorithm}
  The main problem with \ref{algoRSANonUniform} is that the output it
  produces typically is not uniform since the sets $\Set{z \in \R;
    (p,z) \in \RSASET_{x}} \cap \N$ do not necessarily have the same
  cardinality when changing $p$.
\end{fullversion}
To retain uniform selection, we need to select the primes $p$
non-uniformly with the following distribution:
\begin{definition}\label{def:density}
  Let $\RSASET$ be a notion of RSA integers with tolerance~$r$.  For
  every $x \in \R_{>1}$ the associated \emph{cumulative distribution
    function} of $\RSASET_{x}$ is defined as
  $$
  \map[F_{\RSASET_{x}}]{\R}{[0,1]}{y}{\frac{~area\left(\RSASET_{x}
        \cap ([1, y] \times \R) \right)}{~area(\RSASET_{x})}.}
  $$
\end{definition}
In fact we should use the function $
\map[G_{\RSASET_{x}}]{\R}{[0,1]}{y}{\frac{\#\left(\RSASET_{x} \cap
      (([1, y] \cap \mathbb{P}) \times \mathbb{P})
    \right)}{\RSA_{x}},} $ in order to compute the density but
computing $G_{\RSASET_{x}}$ (or its inverse) is tremendously
expensive.  Fortunately, by virtue of \ref{res:psalm23} we know that
$F_{\RSASET_{x}}$ approximates $G_{\RSASET_{x}}$ quite well for
monotone, $[c_1, c_2]$-balanced notions $\RSASET$.  So we use the
function $F_{\RSASET_{x}}$ to capture the distribution properties of a
given notion of RSA integers.  As can be seen by inspection, in
practically relevant examples this function is sufficiently easy to
handle, see \ref{tab:densitiesDeriv}.
\begin{fullversion}
  Using this we modify \ref{algoRSANonUniform} such that each element
  from~$\RSASET_{x}$ is selected almost uniformly at random:
\end{fullversion}
\begin{shortversion}
  We obtain the following algorithm:
\end{shortversion}
\begin{algorithm}{algoRSAFinal}[Generating an RSA integer]
\item A notion $\RSASET$, a bound $x \in \R_{>1}$.
\item An integer $n=p q$ with $(p,q) \in \RSASET_{x}$.
\item \begin{blockuntil}{$y$ prime.}
  \item Select $y$ with distribution $F_{\RSASET_{x}}$ from $\Set{y
      \in \R; \exists z\colon (y,z) \in \RSASET_{x}} \cap \N$.
  \end{blockuntil}
\item $p \gets y$.
\item \begin{blockuntil}{$z$ prime.}
  \item Select $z$ uniformly at random from $\Set{z \in \R; (p,z) \in
      \RSASET_{x}} \cap \N$.
  \end{blockuntil}
\item $q \gets z$.
\item \RETURN $p q$.
\end{algorithm}
As desired, this algorithm generates any pair $(p,q) \in \RSASET_{x}
\cap \left( \PR \times \PR \right)$ with almost the same probability.
In order to generate $y$ with distribution $F_{\RSASET_{x}}$ one can
use inverse transform sampling, see for example \cite{knu98}%
\begin{fullversion}[.\refstepcounter{equation} ]%
  :
  \begin{theorem}[Inverse transform sampling]
    \label{res:invTransSampling}
    Let $F$ be a continuous cumulative distribution function with
    inverse $F^{-1}$ for $u \in [0,1]$ defined by
    $$F^{-1}(u) := ~inf \Set{ x\in\R; F(x)=u }.$$
    If $U$ is uniformly distributed on $[0,1]$, then $F^{-1}(U)$
    follows the distribution $F'$.
  \end{theorem}
  \begin{proof}
    We have $~prob(F^{-1}(U) \leq x) = ~prob(U \leq F(x)) = F(x)$.
  \end{proof}
\end{fullversion}
The expected number of primality tests now is in $\bigO{~ln x}$\begin{fullversion}[.]: If $\RSASET$ is $[c_{1},1]$-balanced then $F_{\RSASET_{x}}(y) = 0$
    as long as $y \leq x^{c_{1}}$.  The exit probability of the first
    loop is $~prob(\text{$y$ prime})$ where $y$ is chosen according to
    the distribution $F_{\RSASET_{x}}'$. Thus
    \begin{align*}
      ~prob(\text{$y$ prime}) &\sim \preint{y}{1}{x}
      \frac{F_{\RSASET_{x}}'(y)}{~ln y} \postint \in
      \left[\frac{1}{~ln x},\frac{1}{c_{1}~ln x}\right]
    \end{align*}
    and we expect $\bigO{~ln x} \cap \bigOmega{ c_{1} ~ln x}$
    repetitions of the upper loop until $y$ is prime.
\end{fullversion}
Of course we have to take into account that for each trial~$u$ an
inverse $F_{\RSASET_{x}}^{-1}(u)$ has to be computed ---~at least
approximately~---, yet this cost is usually negligible compared to a
primality test\begin{fullversion}, see
  \ref{tab:densitiesDeriv}\end{fullversion}.

\subsection{Other constructions}
There are variants around, where the primes are selected differently:
Take an integer randomly from a suitable interval and increase the
result until the first prime is found.  This has the advantage that
the amount of randomness needed is considerably lower and by
optimizing the resulting algorithm can also be made much faster.  The
price one has to pay is that the produced primes will not be selected
uniformly at random: Primes $p$ for which $p-2$ is also prime will be
selected with a much lower probability than randomly selected primes
of a given length. As shown in \cite{bradam93} the output entropy of
such algorithms is still almost maximal and also generators based on
these kind of prime-generators might be used in practice.

\subsection{Comparison}
We have seen that \ref{algoRSALasVegas} and \bare\ref{algoRSAFinal}
are practical uniform generators for any symmetric or antisymmetric
notion.

Note that \ref{algoRSALasVegas} and \bare\ref{algoRSAFinal} may,
however, still produce numbers in a non-uniform fashion: In the last
step of both algorithms a product is computed that corresponds to
either one pair or two pairs in $\RSASET_{x}$.  To solve this problem
we have two choices: Either we replace $\RSASET$ by its symmetric
version~$\mathcal{S}$ defined by $ \mathcal{S}_{x} :=
\Set{(y,z) \in \R_{>1}^{2}; (y,z) \in \RSASET_{x} \vee (z,y) \in
  \RSASET_{x}}, $ or by its, say, top half~$\mathcal{T}$ given by $
\mathcal{T}_{x} := \Set{(y,z) \in \mathcal{S}_{x}; z\geq y} $ before
anything else.

\begin{fullversion}
  It is now relatively simple to instantiate the above algorithms
  using the notions proposed in \ref{sec:NotionsInParticular}: Namely
  for an algorithm following the Las Vegas approach, one simply needs
  to find suitable banner that encloses the desired notion.  In order
  to instantiate \ref{algoRSAFinal} we need to determine the inverse
  of the corresponding cumulative distribution function for the
  respective notion%
  \begin{fullversion}[. We leave further details to the reader and the
    extended version \ifanonymous Anonymous (2011)\else\cite{loenus11a}\fi. ]%
    \ (see \ref{tab:densitiesDeriv}).%
    \begin{table}[t!]
      \centering
      \begin{tabular}{|l|l|c|}
        \hline
        Notion $\RSASET$ & $F'_{\RSASET_{x\mathstrut}}$ & Plot\\
        \hline
        $\RSADMSET$ & 
        \footnotesize $\begin{cases}
          \frac{2 (r^{\mathstrut 2} y^{2} - x)}{x y (r-1) ~ln r}  & \text{if
            $\frac{\sqrt{x}}{r} < y \leq  \sqrt{\frac{x}{r}}$},
          \\
          \frac{2 r (x-y^{2})}{x y (r-1) ~ln r}  & \text{if
            $\sqrt{\frac{x}{r}} < y \leq  \sqrt{x}$},
          \\
          \mathstrut 0 & \text{otherwise}.
        \end{cases}$ & 
        \psset{yunit=25mm, xunit=0.5mm,labelsep=1mm}
        \begin{pspicture}[shift=-2ex](-0.5,0)(50,0.5)
          \psaxes[ticks=none, labels=none]{->}(0,0)(48,0.4)%
          \uput{.5mm}[0]{0}(48,0){$\footnotesize{y}$}%
          \psset{plotpoints=100}%
          \psplot{0}{16}{ 0 %
          }%
          \psplot{16}{22.627416997969522}{ 1 128 div x 2 exp 256 sub
            mul x 2 log mul div%
          }%
          \psplot{22.627416997969522}{32}{ 1 256 div x 2 exp 1024 sub
            mul x 2 log mul div neg }%
          \psplot{32}{48}{ 0 %
          }
        \end{pspicture}
        \\
        \hline $\RSAFIXSET$ & \footnotesize $\begin{cases}
          \frac{\sqrt{r}\left( r^{\frac{1+\sigma}{2}} - 1
            \right)}{\sqrt{x} \left( \sigma r \ln r + r - 2
              r^{\frac{1+\sigma}{2}} + 1 \right)} & \text{if
            $\sqrt{\frac{x}{r}} < y \leq
            \sqrt{\frac{x}{r^{\sigma}}}$},
          \\
          \frac{r \sqrt{x} - \sqrt{r} y}{\sqrt{x} y \left( \sigma r
              ~ln r + r - 2 r^{\frac{1+\sigma}{2}} + 1 \right)} &
          \text{if $\sqrt{\frac{x}{r^{\sigma}}} < y \leq
            \sqrt{r^{\sigma} x}$},
          \\
          0 & \text{otherwise}.
        \end{cases}$ & \psset{yunit=75mm, xunit=0.5mm,labelsep=1mm}
        \begin{pspicture}[shift=-4ex](-0.5,-0.03)(50,0.2)
          \psaxes[ticks=none, labels=none]{->}(0,0)(48,0.13)%
          \psset{plotpoints=100}%
          \uput{.5mm}[0]{0}(48,0){$\footnotesize{y}$}%
          \psplot{0}{22.627416997969522}{ 0 %
          }%
          \psplot{22.627416997969522}{28.319761280092951}{
            0.091536003124457324 }%
          \psplot{28.319761280092951}{35.311032113216946}{ 7.0682249 x
            div 0.15805 sub }%
          \psplot{35.311032113216946}{48}{ 0 %
          }
        \end{pspicture}
        \\
        \hline $\RSAALGSET$ & \footnotesize $\begin{cases} \frac{1}{y
            ~ln r} & \text{if $\frac{\sqrt{x}}{r} < y \leq \sqrt{x}$},
          \\
          0 & \text{otherwise}.
        \end{cases}$ & \psset{yunit=150mm, xunit=.5mm,labelsep=1mm}
        \begin{pspicture}[shift=*](-0.5,-0.02)(50,0.07)
          \psaxes[ticks=none, labels=none]{->}(0,0)(48,0.06)%
          \psset{plotpoints=100}%
          \uput{.5mm}[0]{0}(48,0){$\footnotesize{y}$}%
          \psplot{0}{15.811388300841898}{ 0 %
          }%
          \psplot{15.811388300841898}{31.622776601683796}{ 1 2 div 2
            sqrt mul x div }%
          \psplot{31.622776601683796}{48}{ 0 }%
        \end{pspicture}
        \\
        \hline
      \end{tabular}
      \caption{%
        Non-cumulative density functions with respect to $y$.
      }%
      \label{tab:densitiesDeriv}
    \end{table}
  \end{fullversion}%
  Still \ref{algoRSALasVegas} and \bare\ref{algoRSAFinal} are
  practically uniform generators for any symmetric or antisymmetric
  notion.  Considering run-times we observe that \ref{algoRSAFinal} is
  much faster, but we have to use inverse transform sampling to
  generate the first prime.  Despite the simplicity of the approaches
  some common implementations use corrupted versions of
  \ref{algoRSALasVegas} or \bare\ref{algoRSAFinal} as explained below.
  Essentially, they buy some extra simplicity by relaxing the
  uniformity requirement.
\end{fullversion}

\section{Output entropy}\label{sec:entropy}
The entropy of the output distribution is an important quality measure
of a generator. For primality tests several analyses where performed,
see for example \cite{bradam93} or \cite{joypai06}. For generators of
RSA integers we are not aware of any work in this direction.

Let $\RSASET_x$ be any monotone notion. Consider a generator
$G_{\rho}$ that produces a pair of primes $(p,q) \in \RSASET_x$ with
distribution $\rho$. Seen as random variables, $G_{\rho}$ induces two
random variables $P$ and $Q$ by its first and the second coordinate,
respectively.
The entropy of the generator $G_{\rho}$ is given
by
\begin{equation}
\label{eqn:productEntropy}
H(G_{\rho}) = H(P \times Q) = H(P) + H\Left(Q ; P\Right),
\end{equation}
\begin{shortversion}
  where $H$ denotes the binary entropy and $H\Left(Q ; P\Right)$
  denotes the conditional entropy.
\end{shortversion}
\begin{fullversion}
  where $H$ denotes the entropy and the conditional entropy is given
  by
  $$H\Left(Q ; P\Right) = -\sum_{p \in P} ~prob(P = p) \sum_{q \in Q} ~prob(Q = q \mid P = p) ~log_2(~prob(Q = q \mid P = p)).$$
\end{fullversion}%
If $\rho$ is the uniform distribution $U$ we obtain the maximal
entropy, which we can approximate by \ref{res:psalm23}, namely
$$
H(G_{U}) = ~log_2(\RSA(x)) \approx ~log_2(~area(\RSASET_x)) -
~log_2(~ln x) + 1
$$
with an error of very small order. The algorithms from
\ref{sec:genProperly}, however, return the product $P \cdot Q$. The
entropy of this random variable is at most $H(P \times Q)$ and can be
at most one bit smaller than this:
\begin{align*}
  H(P\cdot Q) &= - \sum\limits_{\substack{n=pq \in \N\\(p,q) \in \RSASET_x}} ~prob(P \cdot Q = n) ~log_2(~prob(P \cdot Q = n))\\
  &\geq - \sum\limits_{(p,q) \in \RSASET_{x}} ~prob(P \times Q = (p,q)) ~log_2(2 ~prob(P \times Q = (p,q))) \\
  &= H(P \times Q) - 1.
\end{align*}
\begin{fullversion}
  One should note here that in real-world implementations the
  generation of the primes might be biased, for example when one uses
  the above mentioned quite natural generator \texttt{PRIMEINC},
  analyzed in \citealp{bradam93}. \texttt{PRIMEINC} chooses an integer
  and then outputs the first prime equal to or larger than this
  number.  Note that \ref{algoRSALasVegas} and \ref{algoRSAFinal} do
  not depend on any prime generator but sample integers until they are
  prime. However, this is not the case in many standards and
  implementations discussed in \ref{sec:concImplem}.

  To estimate the entropy of an RSA generator $G = P \times Q$ when
  employing prime generators $P$ and $Q$ with with entropy-loss at
  most $\epsilon_{P}$ and $\epsilon_{Q}$ then the resulting generator
  will by \ref{eqn:productEntropy} have an entropy-loss of at most
  $\epsilon_P + \epsilon_Q$
  \label{res:inherited-entropy-loss}
  when compared to the same generator employing generators that
  produce primes uniformly at random.
 \end{fullversion}

\begin{fullversion}[Some] Interestingly, some \end{fullversion}%
of the standards and implementations in \ref{sec:concImplem}
(like the standard IEEE 1363-2000 or the implementation of \texttt{GNU
  Crypto}) \begin{fullversion}still\end{fullversion} \emph{do not} generate every possible outcome with the same
probability, even if uniform prime generators are employed: Namely, if one selects the prime $p$  uniformly at random and afterwards the prime $q$  uniformly
at random from an appropriate interval then this might be a non-uniform
selection process if for some choices of $p$ there are less
choices for $q$.

If in general the probability distribution $\rho$ is close to the
uniform distribution, say $\rho(p,q) \in [2^{-\epsilon}, 2^{\epsilon}]
\frac{1}{\RSA(x)}$ for some fixed $\epsilon \in \R_{> 0}$, then the
entropy of the resulting generator $G_{\rho}$ can be estimated
as \begin{fullversion}follows:
  \begin{align*}
    H(G_{\rho}) &= -\sum_{(p,q) \in \RSASET_x} \rho(p,q) ~log_2(\rho(p,q))\\
    &\in \sum_{(p,q) \in \RSASET_x} \rho(p,q) [~log_2(\RSA(x)) -
    \epsilon, ~log_2(\RSA(x)) + \epsilon]
    \\
    &= [H(G_{U}) - \epsilon, H(G_{U}) + \epsilon]
  \end{align*}
  and since the entropy of the uniform distribution is maximal, this
  implies that \end{fullversion}
$$
H(G_{U}) - \epsilon \leq H(G_{\rho}) \leq H(G_U).
$$

\section{Complexity theoretic considerations}\label{sec:algoConsid}
\def\arealargenesswitness{k} \def\successwitness{s}
\def\attackeralgo{{F}}

We are about to reduce factoring products of two comparatively equally
sized primes to the problem of factoring integers generated from a
sufficiently large notion. As far as we know there are no similar
reductions in the literature.

We consider finite sets $M \subset \N \times \N$, in our situation we
actually have only prime pairs.  The multiplication map $\mu_{M}$ is
defined on $M$ and merely multiplies, that is, $\map[\mu_{M}] {M}{\N}
{(y,z)}{y\cdot z}$.  The random variable $U_{M}$ outputs uniformly
distributed values from $M$.  An attacking algorithm $\attackeralgo$
gets a natural number $\mu_{M}(U_{M})$ and attempts to find factors
inside $M$.  Its success probability
\begin{align}
  \label{def:succ}
  ~succ_{\attackeralgo}(M) &= \prob{ \quad \attackeralgo(
    \mu_{M}(U_{M}) ) \;\in\; \mu_{M}^{-1}( \mu_{M}(U_{M}) ) \quad }
\end{align}
measures its quality in any fixed-size scenario.
\begin{fullversion}
  We call a countably infinite family $\mathcal{C}$ of finite sets of
  pairs of natural numbers \emph{hard to factor} iff for any
  probabilistic polynomial time algorithm $\attackeralgo$ and any
  exponent $\successwitness$ for all but finitely many $M \in
  \mathcal{C}$ the success probability $~succ_{\attackeralgo}(M) \leq
  ~ln^{-\successwitness} x$ where $x = ~max \mu_{M}(M)$.  In other
  words: the success probability of any probabilistic polynomial time
  factoring algorithm on a number chosen uniformly from $M \in
  \mathcal{C}$ is negligible.  That is equivalent to saying that the
  function family $(\mu_{M})_{M\in\mathcal{C}}$ is one-way.

\end{fullversion}
Integers generated from a notion $\RSASET$ are \emph{hard to factor}
iff
\begin{fullversion}
  for any sequence $(x_i)_{i \in \N}$ tending to infinity the family
  $(\RSASET_{x_i} \cap (\PR\times\PR))_{i \in \N}$ is hard to factor.
  This definition is equivalent to the requirement that
\end{fullversion}
for all probabilistic polynomial time machines $F$, all $s \in \N$,
there exists a value $x_0 \in \R_{>1}$ such that for any $x > x_0$ we
have $~succ_F(\RSASET_{x}) \leq ~ln^{-s} x$.
\begin{fullversion}
  Since $\R$ is first-countable, both definitions are actually equal.
  This can be easily shown by considering the functions $\map[g_{s,
    F}]{\R_{>1}}{\R}{x}{~succ_F(\RSASET_{x}) \cdot ~ln^{s} x}$.
  The first definition says that each function $g_{s,F}$ is
  sequentially continuous (after swapping the initial universal
  quantifiers).  The second definition says that each function
  $g_{s,F}$ is continuous.  In first-countable spaces sequentially
  continuous is equivalent to continuous.
\end{fullversion}

\noindent For any polynomial $f$ we define the set $R_{f} =
\Set{(m,n)\in\N; m \leq f(n) \land n \leq f(m) }$ of $f$-related
positive integer pairs.  Denote by $\PR^{(m)}$ the set of $m$-bit
primes.  We can now formulate the basic assumption:
\begin{assumption}[Intractability of factoring]
  \label{ass:factoring-pairs}
  For any unbounded positive polynomial~$f$ integers from the
  \emph{$f$-related prime pair family} $(\PR^{(m)}\times
  \PR^{(n)})_{(m,n)\in R_{f}}$ are hard to factor.
\end{assumption}
This is exactly the definition given by
\cite{gol01}. 
\ifdraft
\endnote{Call $f$ hard iff the family of $f$-related prime pairs is
  hard to factor.  Clearly, $x^{2}$ hard implies that $a x$ is hard
  for any $a>0$.  Does $x^{2}$ hard imply that all polynomials are
  hard?  Does $a x$ hard imply that all polynomials are hard?  Does $2
  x$ hard imply that all polynomials are hard? (No!?)}
\fi Note that this assumption implies that factoring in general is
hard, and it covers the supposedly hardest factoring instances.
Now we are ready to state that integers from all relevant notions are
hard to factor.

\begin{theorem}
  \label{res:RSAalgo-difficult}
  Let $~ln r \in \bigOmega{\frac{1-2c_1}{~ln^l x}}$ for some $l$ and $\RSASET$ be a
  \begin{fullversion}[ ]%
    piecewise%
  \end{fullversion}
  monotone, $[c_{1}, c_{2}]$-balanced notion for RSA integers of
  tolerance~$r$, where $c_1$ is bounded away from zero with growing
  $x$, and $\RSASET$ has large area, namely, for some $k$ and large
  $x$ we have $~area \RSASET_{x} \geq
  \frac{x}{~ln^{\arealargenesswitness} x}$.  Assume that factoring is
  difficult in the sense of
  \ref{ass:factoring-pairs}.
  Then integers from the notion~$\RSASET$ are hard to
  factor. \ifshortversion \qed \fi
\end{theorem}
\begin{fullversion}
  Actually, under the given conditions \ref{ass:factoring-pairs} can
  be weakened: we only need that integers from the family of
  \emph{linearly} related prime pairs are hard to factor.  There is a
  tradeoff between the strength of the needed assumption on factoring
  and the assumption on $c_{1}$.  If we relax the restriction on $c_1$
  in the statement of the theorem to the requirement that $c_1^2 ~ln
  x$ tends to infinity with growing $x$, we need that integers from
  the family of \emph{quadratically} related prime pairs are hard to
  factor.
\end{fullversion}%
\begin{proof}
  Assume that we have an algorithm $\attackeralgo$ that factors
  integers generated uniformly from the notion $\RSASET$.  Our goal is
  to prove that this algorithm also factors certain polynomially
  related prime pairs successfully.  In other words: its existence
  contradicts the assumption that factoring in the form of
  \ref{ass:factoring-pairs} is difficult.
  
  By assumption, there is an exponent $\successwitness$ so that for
  any $x_0$ there is $x > x_0$ such that the assumed algorithm
  $\attackeralgo$ has success probability
  $~succ_{\attackeralgo}(\RSASET_{x}) \geq ~ln^{-\successwitness} x$
  on inputs from $\RSASET_{x}$.  We are going to prove that for each
  such $x$ there exists a pair $(m_{0},n_{0})$, the entries both from
  the interval $[c_{1} ~ln x - ~ln 2, c_{2} ~ln x + ~ln 2]$, such that
  $\attackeralgo$ executed with an input from $~image
  \mu_{\mathbb{P}^{m_{0}},\mathbb{P}^{n_{0}}}$ still has success
  probability at least $~ln^{-(\successwitness+\arealargenesswitness)}
  x$.  By the interval restriction, $m_{0}$ and $n_{0}$ are
  polynomially (even linearly) related, namely $m_{0} < \frac{2
    c_{2}}{c_{1}} n_{0}$ and $n_{0} < \frac{2 c_{2}}{c_{1}} m_{0}$ for
  large $x$.  By the assumption on $c_1$ the fraction $\frac{2
    c_{2}}{c_{1}}$ is bounded independent of $x$.  So that contradicts
  \ref{ass:factoring-pairs}.

  First, we cover the set $\RSASET_{x}$ with small rectangles.  Let
  $S_{m,n} := \PR^{(m)} \times \PR^{(n)}$ and $ I_{x} := \Set{(m,n)
    \in \N^{2}; S_{m,n} \cap \RSASET_{x} \neq \emptyset} $ then
  \begin{align}
    \RSASET_{x} \cap \PR^{2} &\subseteq \biguplus_{(m,n) \in I_{x}}
    S_{m,n} =: S_{x}.
  \end{align}

  Next we give an upper bound on the number $\#S_{x}$ of prime pairs
  in the set $S_{x}$ in terms of the number $\RSA(x)$ of prime pairs
  in the original notion: First, since each rectangle $S_{m,n}$
  extends by a factor $2$ along each axis we overshoot by at most that
  factor in each direction, that is, we have for $c_1' = c_1 -
  (1+2c_1) \frac{~ln 2}{~ln x}$ and all $x \in \R_{>1}$
  \begin{align*}
    S_{x} &\subset \RSAMAXSET[16r, c_{1}']_{4x} = \Set{(y,z) \in
      \R^{2}; y,z \geq \frac{1}{2} x^{c_{1}} \Land \frac{x}{4r} < y z
      \leq 4x }.
  \end{align*}
  Provided $x$ is large enough we can guarantee by
  \begin{shortversion}[\ref{res:order-rsamax} that]
    Theorem 5.2 from the extended version (similar to
    \ref{res:psalm23}) that
  \end{shortversion}
  \begin{align*}
    \#S_{x} &\leq \RSAMAX[16r, c_{1}']( 4 x ) \leq \frac{8
      x}{c_{1}'^{2} ~ln x}.
  \end{align*}
  On the other hand side we apply \ref{res:psalm23} for the notion
  $\RSASET_{x}$ and use that $\RSASET_{x}$ is large by assumption.
  Let $c = ~max\left( 2c_{2}-1, 1-2c_{1} \right)$.  Then we obtain for
  large $x$ with some $e_{\RSASET}(x) \in \bigO{c_1^{-1}
    x^{\frac{3+c}{4}}}$.
  \begin{align*}
    \RSA(x) &\geq \frac{~area(\RSASET_{x})}{c_{2}^{2} ~ln^{2} x} -
    e_{\RSASET}(x) \geq \frac{x}{2 c_{2}^{2}
      ~ln^{\arealargenesswitness+2} x}.
  \end{align*}
  Together we obtain
  \begin{align}
    \label{res:cover-ratio}
    \frac{\RSA(x)}{\#S_{x}} &\geq \frac{c_{1}'^{2}}{16 c_{2}^{2}
      ~ln^{\arealargenesswitness+1} x} \geq
    ~ln^{-(\arealargenesswitness+2)} x
  \end{align}
  
  By assumption we have $~succ_{\attackeralgo}(\RSASET_{x}) \geq
  ~ln^{-\successwitness} x$ for infinitely many values $x$. Thus
  $\attackeralgo$ on an input from $S_{x}$ still has large success
  even if we ignore that $\attackeralgo$ might be successful for
  elements on $S_{x} \setminus \RSASET_{x}$,
  $$
  ~succ_{\attackeralgo}(S_{x}) \geq ~succ_{\attackeralgo}(\RSASET_{x})
  \frac{\#\RSASET(x)}{\#S_{x}} \geq
  ~ln^{-(\arealargenesswitness+\successwitness+2)} x.
  $$
  Finally choose $(m_{0},n_{0}) \in I_{x}$ for which the success of
  $\attackeralgo$ on $S_{m_{0},n_{0}}$ is maximal
  .  Then $~succ_{\attackeralgo}(S_{m_{0},n_{0}}) \geq
  ~succ_{\attackeralgo}(S_{x}) $.  Combining with the previous we
  obtain that for infinitely many $x$ there is a pair $(m_{0},n_{0})$
  where the success $~succ_{\attackeralgo}( S_{m_{0},n_{0}})$ of
  $\attackeralgo$ on inputs from $S_{m_{0},n_{0}}$ is still larger
  than inverse polynomial: $~succ_{\attackeralgo}( S_{m_{0},n_{0}})
  \geq ~ln^{-(\arealargenesswitness+\successwitness+2)} x$.

  For these infinitely many pairs $(m_{0},n_{0})$ the success
  probability of the algorithm $\attackeralgo$ on $S_{m_{0},n_{0}}$ is
  at least $~ln^{-(\arealargenesswitness+\successwitness+2)} x$
  contradicting the hypothesis.
\end{proof}

All the specific notions that we have found in the literature fulfill
the criterion of \ref{res:RSAalgo-difficult}.  Thus if factoring is
difficult in the stated sense then each of them is invulnerable to
factoring attacks. Note that the above reduction still works if the
primes $p,q$ are due to the side condition $~gcd((p-1)(q-1), e) = 1$
for a fixed integer $e$ (see \ref{res:coprimalityE}). We suspect that 
this is also the case when one employes safe primes.

%
\cite{mih01} shows a theorem which is in some respect more general
than our considerations and seems to imply our
\ref{res:RSAalgo-difficult}.  To that end one has to show that one of
$\RSASET_{x}$ and $\PR^{(m)}\times \PR^{(n)}$, with suitably chosen
$(m,n)$ depending on $x$, is `polynomially dense' in the other.  The
result would be more general since also the used distribution on
$\RSASET_{x}$, rather than being uniform, is allowed to be
`polynomially bounded'.  
Our proof is of a different nature and thus may well be of independent
interest.  Also it may be adapted to polynomially bounded
distributions on $\RSASET$.



\section{Impact on standards and implementations}
\label{sec:concImplem}
In order to get an understanding of the common implementations, it is
necessary to consult the main standard on RSA integers, namely the
standard PKCS\#1 \citep{jonkal03}. However, one cannot find \emph{any}
requirements on the shape of RSA integers there. Interestingly, they
even allow more than two factors for an RSA modulus. Also the standard
ISO 18033-2 \citep{iso180332} does not give any details besides the
fact that it requires the RSA integer to be a product of two different
primes of similar length.
\begin{fullversion}
  A more precise standard is set by the German Bundesnetzagentur
  \citep{woh08}. They do not give a specific algorithm, but at least
  require that the prime factors are not too small and not too close
  to each other.  We will now analyze several standards which give a
  concrete algorithm for generating an RSA integer. In particular, we
  consider the standard of the RSA foundation \citep{rsa00}, the IEEE
  standard 1363 \citep{araarn00}, the NIST standard FIPS 186-3
  \citep{fips1863}\ifshortversion and\else, \fi the standard ANSI
  X9.44 \citep{ansi944}
  \begin{fullversion}[.] and the standard resulting from the European
    NESSIE project \citep{prebir03}.
  \end{fullversion}
\end{fullversion}

\ifshortversion\pagebreak\fi
\subsection{RSA-OAEP}
The \cite{rsa00} describe the following variant:
\iffullversion\pagebreak[3]\else\pagebreak[3]\fi
\hanghere[-\baselineskip]{%
  \drawnotion[]{6}{2}{%
    \RSAFIXdraw{0} }%
}%
\begin{algorithm}{algoRSAKG}[Generating an RSA number for RSA-OAEP and
  variants]
  \advance\linewidth-13\psxunit
\item A number of bits $k$, the public exponent $e$.
\item A number $n = p q$.
\item Pick $p$ from $\left[\floor{2^{(k-1)/2}}+1,
    \ceil{2^{k/2}}-1\right] \cap \mathbb{P}$ such that\\$~gcd(e,p-1)
  =1$.
\item Pick $q$ from $\left[\floor{2^{(k-1)/2}}+1,
    \ceil{2^{k/2}}-1\right] \cap \mathbb{P}$ such that\\$~gcd(e,q-1)
  =1$.
\item \RETURN $p q$.
\end{algorithm}
This will produce a number from the interval $[ 2^{k-1}+1, 2^{k}-1 ]$
and no cutting off.  The output entropy is maximal.  So this
corresponds to the notion $\RSAFIXSET[2,0]$ generated by
\ref{algoRSAFinal}. The standard requires an expected number of $k ~ln
2$ primality tests if the $~gcd$ condition is checked first. Otherwise
the expected number of primality tests increases to
$\frac{\phi(e)}{\phi_1(e)} \cdot k ~ln 2$, see \ref{def:phi1}.  We
will in the following always mean by the above notation that the
second condition is checked first and afterwards the number is tested
for primality.  For the security \ref{res:RSAalgo-difficult} applies.

\subsection{IEEE}
IEEE standard 1363-2000, Annex A.16.11 \citep{araarn00} introduces our
algorithmic proposal:
\iffullversion\pagebreak\fi
\hanghere{%
  \newpsstyle{notionsplit}{linestyle=none}%
  \drawnotion[]{6}{2}{%
    \RSAALGVARdraw{0.5}}}%
\begin{algorithm}{algoIEEEP1363-A16-12}[Generating an RSA number, IEEE
  1363-2000]
  \advance\linewidth-13\psxunit
\item A number of bits $k$, the odd public exponent $e$.
\item A number $n = p q$.
\item Pick $p$ from $\left[ 2^{\quo{k-1}{2}},
    2^{\quo{k+1}{2}}-1\right] \cap \mathbb{P}$ such that
  \\$~gcd(e,p-1) =1$.
\item Pick $q$ from $\left[\floor{\frac{2^{k-1}}{p}+1},
    \floor{\frac{2^{k}}{p}}\right] \cap \mathbb{P}$ such that
  \\$~gcd(e,q-1) =1$.
\item \RETURN $p q$.
\end{algorithm}
Since the resulting integers are in the interval $[2^{k-1}, 2^{k}-1]$
this standard follows $\RSAALGVARSET[2, 1/2]$ generated by a corrupted
variant of \ref{algoRSAFinal} using an expected number of $k \ln 2$
primality tests like the RSA-OAEP standard. The notion it implements
is neither symmetric nor antisymmetric. The selection of the integers
is \emph{not} done in a uniform way, since the number of possible $q$
for the largest possible $p$ is roughly half of the corresponding
number for the smallest possible $p$. Since the distribution of the
outputs is close to uniform, we can use the techniques from
\ref{sec:entropy} to estimate the output entropy to find that the
entropy-loss is less than 0.69 bit. The (numerically approximated)
values in \ref{tab:standards} gave an actual entropy-loss of
approximately $0.03$ bit.

\subsection{NIST}
We will now analyze the standard FIPS 186-3 \citep{fips1863}. In
Appendix B.3.1 of the standard one finds the following algorithm:
\hanghere{%
  \drawnotion[]{6}{2}{%
    \RSAFIXdraw{0} }%
}%
\begin{algorithm}{algoFIPS186-3}[Generating an RSA number, FIPS186-3]
  \advance\linewidth-13\psxunit
\item A number of bits $k$, a number of bits $l < k$, the odd public
  exponent $2^{16} < e < 2^{256}$.
\item A number $n = p q$.
\item Pick $p$ from $\left[ \sqrt{2}\, 2^{k/2-1}, 2^{k/2}-1\right]
  \cap \mathbb{P}$ such that\\$~gcd(e,p-1) =1$ and $p \pm 1$ has a
  prime factor with at least $l$ bits.
\item Pick $q$ from $\left[ \sqrt{2}\, 2^{k/2-1}, 2^{k/2}-1\right]
  \cap \mathbb{P}$ such that\\$~gcd(e,p-1) =1$ and $q \pm 1$ has a
  prime factor with at least $l$ bits and $|p-q| > 2^{k/2-100}$.
\item \RETURN $p q$.
\end{algorithm}
In the standard it is required that the primes $p$ and $q$ shall be
either provably prime or at least probable primes.  The mentioned
large (at least $l$-bit) prime factors of $p \pm 1$ and $q \pm 1$ have
to be provable primes. We observe that also in this standard a variant
of the notion $\RSAFIXSET[2,0]$ generated by \ref{algoRSAFinal} is
used.  The output entropy is thus maximal. However, we do not have any
restriction on the parity of $k$, such that the value $k/2$ is not
necessarily an integer. Another interesting point is the restriction
on the prime factors of $p\pm 1 ,q \pm 1$. Our notions cannot directly
handle such requirements, but this may possibly be
achieved by appropriately modifying the prime number densities in the
proof of \ref{res:psalm23}.

The standard requires an expected number of slightly more than $k ~ln
2$ primality tests. It is thus slightly less efficient than the
RSA-OAEP standard. For the security the remarks from the end of
\ref{sec:algoConsid} apply.

\subsection{ANSI}
The ANSI X9.44 standard \citep{ansi944}, formerly part of ANSI X9.31,
requires strong primes for an RSA modulus. Unfortunately, we could not
access ANSI X9.44 directly and are therefore referring to ANSI
X9.31-1998.  Section 4.1.2 of the standard requires that
\begin{itemize}
\item $p-1$, $p+1$, $q-1$, $q+1$ each should have prime factors $p_1$,
  $p_2$, $q_1$, $q_2$ that are randomly selected primes in the range
  $2^{100}$ to $2^{120}$,
\item $p$ and $q$ shall be the first primes that meet the above, found
  in an appropriate interval, starting from a random point,
\item $p$ and $q$ shall be different in at least one of their first
  100 bits.
\end{itemize}
The additional restrictions are similar to the ones required by NIST.
This procedure will have an output entropy that is close to maximal
(see \ref{sec:entropy}).
\begin{fullversion}[\refstepcounter{equation}]
  \subsection{NESSIE}
  The European NESSIE project gives in its security report
  \citep{prebir03} a very similar algorithm:
  \hanghere{%
    \drawnotion[]{6}{4}{%
      \RSAFIXdraw{0} }%
  }%
  \begin{algorithm}{algoNESSIE}[Generating an RSA number, NESSIE
    standard]
    \advance\linewidth-13\psxunit
  \item A number of bits $l$, the odd public exponent $e$.
  \item A number $n = p q$.
  \item Pick $p$ from $\left[ 2^{l-1}, 2^{l}-1\right] \cap \mathbb{P}$
    such that\\$~gcd(e,p-1) =1$.
  \item Pick $q$ from $\left[ 2^{l-1}, 2^{l}-1\right]\cap \mathbb{P}$
    such that\\$~gcd(e,q-1) =1$.
  \item \RETURN $p q$.
  \end{algorithm}
  The resulting integer $n$ is from the interval $[2^{2 l -2},
  2^{2 l}-1]$ and thus corresponds to the fixed-bound notion
  $\RSAFIXSET[4,0]$ generated by \ref{algoRSAFinal}.  The output
  entropy is thus maximal. Note the difference to the standard of the
  RSA foundation: Besides the fact, that in the standard of the RSA
  laboratories some sort of rounding is done, the security parameter
  $l$ is treated differently: While for the RSA foundation the
  security parameter describes the (rough) length of the output, in
  the NESSIE proposal it denotes the size of the two prime factors.
  For comparison let $k = 2 l$.
  The standards requires an expected number of $2k ~ln 2$ primality
  tests. It is thus as efficient as the RSA-OAEP standard. For the
  security \ref{res:RSAalgo-difficult} applies.
\end{fullversion}

\subsection{OpenSSL}
We now turn to implementations: For \texttt{OpenSSL}
\citep{openssl09}, we refer to the file \url{rsa_gen.c}.
Note that in the configuration the routine used for RSA integer
generation can be changed, while the algorithm given below is the
standard one.  \texttt{OpenSSH} \citep{openssh09} uses the same
library. Refer to the file \url{rsa.c}.
We have the following algorithm:
\iffullversion\pagebreak[3]\else\pagebreak[3]\fi%
\hanghere{%
  \drawnotion[]{6.6}{4}{%
    \pspolygon[style=notionarea]%
    (-0.5,0.5)(0.5,0.5)(0.5,1.5)(-0.5,1.5)%
  }%
}%
\begin{algorithm}{algoRSASSL}[Generating an RSA number in
  \texttt{OpenSSL}]
  \advance\linewidth-13\psxunit
\item A number of bits $k$.
\item A number $n = p q$.
\item Pick $p$ from $\left[2^{\quo{k-1}{2}}, 2^{\quo{k+1}{2}}-1\right]
  \cap \mathbb{P}$.
\item Pick $q$ from $\left[2^{\quo{k-3}{2}}, 2^{\quo{k-1}{2}}-1\right]
  \cap\mathbb{P}$.
\item \RETURN $p q$.
\end{algorithm}
This is nothing but a rejection-sampling method with no rejections of a notion similar to
the fixed-bound notion $\RSAFIXSET[4,0]$ generated by
\ref{algoRSALasVegas}.  The output entropy is thus maximal.  The
result the algorithm produces is always in $[2^{k-2}, 2^{k}-1]$.  It
is clear that this notion is antisymmetric and the factors are on
average a factor $2$ apart of each other. The implementation runs in
an expected number of $k ~ln 2$ primality tests. The public exponent
$e$ is afterwards selected such that~$~gcd((p-1)(q-1), e) = 1$. It is
thus slightly more efficient than the RSA-OAEP standard. For the
security \ref{res:RSAalgo-difficult} applies.

\begin{fullversion}[\refstepcounter{equation}]
  \subsection{Openswan}
  In the open source implementation \texttt{Openswan} of the IPsec
  protocol \citep{openswan09} one finds a rejection-sampling method
  that is actually implementing the notion $\RSAFIXSET[4,0]$ generated
  by a variant of \ref{algoRSALasVegas}. We refer to the function
  \texttt{rsasigkey} in the file
  \url{rsasigkey.c}:
  \hanghere{%
    \drawnotion[]{6}{4}{%
      \RSAFIXdraw{0} }%
  }%
  \begin{algorithm}{algoRSAOpenSwan}[Generating an RSA number in
    \texttt{Openswan}]
    \advance\linewidth-13\psxunit
  \item A number of bits $k$.
  \item A number $n = p q$.
  \item Pick $p$ from $\left[2^{\quo{k-2}{2}}, 2^{\quo{k}{2}}-1\right]
    \cap \mathbb{P}$.
  \item Pick $q$ from $\left[2^{\quo{k-2}{2}}, 2^{\quo{k}{2}}-1\right]
    \cap \mathbb{P}$.
  \item \RETURN $p q$.
  \end{algorithm}
  Note that here the notion \emph{is} actually symmetric. However
  still the uniformly at random selected number $p q$ will not always
  have the same length.  The implementation runs in an expected number
  of $k ~ln 2$ primality tests and output entropy is maximal.  Again
  the public exponent $e$ is afterwards selected such that
  $~gcd((p-1)(q-1), e) = 1$. It is thus as efficient as the RSA-OAEP
  standard. For the security \ref{res:RSAalgo-difficult} applies.
\end{fullversion}

\subsection{GnuPG}
Also \texttt{GnuPG} \citep{gnupg09} uses rejection-sampling of the
fixed-bound notion $\RSAFIXSET[2,1]$ generated by a variant of
\ref{algoRSALasVegas}, implying that the entropy of its output
distribution is maximal.
\iffullversion\pagebreak[3]\else\pagebreak[3]\fi%
\hanghere{%
  \drawnotion[]{6}{2}{%
    \newpsstyle{notionsplit}{linestyle=none} \rput(-1,-1){%
      \psset{unit=2\psunit}%
      \newpsstyle{notionarea}{style=bannerarea} \RSAFIXdraw{0}%
    }%
    \RSAFIXdraw{1} }%
}%
\begin{algorithm}{algoRSAGnuPG}[Generating an RSA number in
  \texttt{GnuPG}]
  \advance\linewidth-13\psxunit
\item A number of bits $k$.
\item A number $n = p q$.
\item \begin{blockuntil}{$~len(p q) = 2\ceil{k/2}$}
  \item Pick $p$ from $\left[2^{\quo{k-1}{2}},
      2^{\quo{k+1}{2}}-1\right] \cap \mathbb{P}$.
  \item Pick $q$ from $\left[2^{\quo{k-1}{2}},
      2^{\quo{k+1}{2}}-1\right] \cap \mathbb{P}$.
  \end{blockuntil}
\item \RETURN $p q$.
\end{algorithm}
The hatched region in the picture above shows the possible outcomes
that are discarded. We refer here to the file \url{rsa.c}.
The algorithm is given in the function \url{generate_std} and produces
always numbers with either $k$ or $k+1$ bits depending on the parity
of $k$.  Note that the generation procedure indeed first selects
primes before checking the validity of the range. This is of course a
waste of resources, see \ref{sec:genProperly}.

The implementation runs in an expected number of roughly $2.589 \cdot
(k+1) ~ln 2$ primality tests. It is thus less efficient than the RSA
OAEP standards.  Like in the other so far considered implementations,
the public exponent $e$ is afterwards selected such that
$~gcd((p-1)(q-1), e) = 1$.  For the security
\ref{res:RSAalgo-difficult} applies.

\begin{table}[t!]
  \centering
  \def\elossa#1{\raisebox{0pt}{{\footnotesize#1}}}
  \def\eloss#1{\raisebox{3pt}{{\footnotesize#1}}}
  \let\bugfix\relax
  \begin{tabular}{|c|c|r|r|r|c|c|}
    \hline
    Standard & \multirow{2}{*}{Notion} & \multicolumn{3}{|c|}{Entropy
      \elossa{(entropy loss)}} & \multirow{2}{*}{Remarks}\\
    \cline{1-1}
    Implementation & & $k=768$ & $k=1024$ & $k=2048$ & \\
    \hline
    PKCS\#1   & \multirow{3}{*}{Undefined} & \multirow{3}{*}{---} & \multirow{3}{*}{---} & \multirow{3}{*}{---} &  \multirow{3}{*}{---}\\
    ISO 18033-2 &  &  &  &  &\\
    ANSI X9.44  &  &  &  &  &\\
    FIPS 186-3 & $\RSAFIXSET[2,0]$ & $\lesssim 747.34$\bugfix & $\lesssim 1002.51$\bugfix & $\lesssim 2024.51$\bugfix   & strong primes\\
    & & \eloss{($\gtrsim 0$ \permil)} & \eloss{($\gtrsim 0$ \permil)} & \eloss{($\gtrsim 0$ \permil)} & \\
    RSA-OAEP  &  $\RSAFIXSET[2,0]$ & $747.34$ & 1002.51 & 2024.51 & --- \\
    & & \eloss{($0$ \permil)} & \eloss{($0$ \permil)} & \eloss{($0$ \permil)} & \\
    IEEE 1363-2000~ &  $\RSAALGVARSET[2,\frac12]$ & 749.33 & 1004.50 & 2026.50 & non-uniform \\
    & & \eloss{($0.04$ \permil)} & \eloss{($0.03$ \permil)} & \eloss{($0.01$ \permil)} & \\
    \iffullversion
    NESSIE     & $\RSAFIXSET[4,0]$ & 749.89 & 1005.06 & 2027.06  & ---\\
    & & \eloss{($0$ \permil)} & \eloss{($0$ \permil)} & \eloss{($0$ \permil)} & \\
    \fi
    \hline
    \texttt{GNU Crypto} &  $\RSAFIXSET[2\strut,1]$ & 747.89 & 1003.06 & 2025.06 &  non-uniform\\
    & & \eloss{($0.84$ \permil)} & \eloss{($0.62$ \permil)} & \eloss{($0.31$ \permil)} & \\
    \texttt{GnuPG}    & $\RSAFIXSET[2,1]$ & 748.52 &  1003.69 & 2025.69  & ---\\
    & & \eloss{($0$ \permil)} & \eloss{($0$ \permil)} & \eloss{($0$ \permil)} & \\
    \texttt{OpenSSL}    & $\cong \RSAFIXSET[4,0]$ & 749.89 & 1005.06 & 2027.06  & ---\\
    & & \eloss{($0$ \permil)} & \eloss{($0$ \permil)} & \eloss{($0$ \permil)} & \\
    \iffullversion
    \texttt{Openswan}   & $\RSAFIXSET[4,0]$ & 749.89 & 1005.06 & 2027.06  & ---\\
    & & \eloss{($0$ \permil)} & \eloss{($0$ \permil)} & \eloss{($0$ \permil)} & \\
    \fi
    \hline
  \end{tabular}
  \vspace{1mm}
  \caption{Overview of various standards and implementations.  The
    numbers in parentheses give the entropy loss for each algorithm in
    per mille. As explained in the text, the entropy of the standards
    is sightly smaller than the values given due to the fixed public
    exponent $e$.  FIPS 186-3 has a small entropy loss because of the
    requirement of strong primes.  Generators based on nonuniform
    prime generation suffer extra entropy loss, see page
    \pageref{res:inherited-entropy-loss}.
  }\label{tab:standards}
\end{table}

\subsection{GNU Crypto}
The \texttt{GNU Crypto} library \citep{gnucrypto09} generates RSA
integers the following way.  Refer here to the function
\texttt{generate} in the file \url{RSAKeyPairGenerator.java}.
\hanghere{%
  \drawnotion[]{6}{2}{%
    \RSAFIXdraw{1}%

    \psline[linewidth=1.2pt]{->}(0.8,0.1)(-0.6,0.8)
  }%
}%
\begin{algorithm}{algoRSAGnuCrypto}[Generating an RSA number in
  \texttt{GNU Crypto}]
  \advance\linewidth-13\psxunit
\item A number of bits $k$.
\item A number $n = p q$.
\item Pick $p$ from $\left[2^{\quo{k-1}{2}}, 2^{\quo{k+1}{2}}-1\right]
  \cap \mathbb{P}$.
\item \begin{blockuntil}{$~len(p q) = k$ and $q \in \mathbb{P}$.}
  \item Pick $q$ from $\left[2^{\quo{k-1}{2}},
      2^{\quo{k+1}{2}}-1\right]$.
  \end{blockuntil}
\item \RETURN $p q$.
\end{algorithm}
The arrow in the picture points to the results that will occur with
higher probability.  Also here the notion $\RSAFIXSET[2,1]$ is used,
but the generated numbers will not be uniformly distributed, since for
a larger $p$ we have much less choices for $q$. Since the distribution
of the outputs is not close to uniform, we could only compute the
entropy for real-world parameter choices numerically (see
\ref{tab:standards}). For all choices the loss was less than 0.63 bit.
The implementation is as efficient as the RSA-OAEP standard.

\begin{fullversion}
  The Free Software Foundation provides \texttt{GNU Classpath}, which
  generates RSA integers exactly like the \texttt{GNU Crypto} library,
  i.e. following $\RSAFIXSET[2,1]$. We refer to the source file named
  \url{RSAKeyPairGenerator.java}.  As in the other so far considered
  implementations the public exponent $e$ is  randomly
  selected afterwards such that $~gcd((p-1)(q-1), e) = 1$.  Like in the IEEE and
  the ANSI standard this does not impose practical security risks, but
  it does not meet the requirement of uniform selection of the
  generated integers.
\end{fullversion}

\begin{fullversion}
  \subsection{Summary}
  It is striking to observe that not \emph{a single} analyzed
  implementation follows one of the standards described above.  The
  only standards all implementations are compliant to are the
  standards PKCS\#1 and ISO~18033-2, which themselves do not specify
  anything related to the integer generation routine.  We found that
  also the requirements from the algorithm catalog of the German
  Bundesnetzagentur \citep{woh08} are not met in a single considered
  implementation, since it is never checked whether the selected
  primes are too close to each other.  The implementation that almost
  meets the requirements is the implementation of \texttt{OpenSSL}.
  Interestingly there are standards and implementations around that
  generate integers non-uniformly.  Prominent examples are the IEEE
  and the ANSI standards and the implementation of the \texttt{GNU
    Crypto} library.  This does not impose practical security issues,
  but it violates the condition of uniform selection.
\end{fullversion}

\section{Conclusion}
We have seen that there are various definitions for RSA integers,
which result in substantially differing standards. We have shown that
the concrete specification does not essentially affect the
(cryptographic) properties of the generated integers: The entropy of
the output distribution is always almost maximal, generating those
integers can be done efficiently, and the outputs are hard to factor
if factoring in general is hard in a suitable sense.  It remains open
to incorporate strong primes into our model. Also a tight bound for
the entropy of non-uniform selection is missing if the distribution is
not close to uniform.



\ifanonymous\else
\section*{Acknowledgements}
This work was funded by the B-IT foundation and the state of
North~Rhine-Westphalia.  \fi

\bibliography{journals,refs,lncs}

\end{document}

%% file: rsa-standards-density.tex
%
%
%
\psset{unit=7mm,labelsep=1mm}
\begin{pspicture}(-0.5,-0.5)(10,10)
  \def\parplot#1#2#3{%
    \parametricplot{#2}{#3}{%
      t 2 exp %
      2.718281828 t 2 mul #1 sub exp %
      sub %
      dup 0 ge %
      {}%
      { pop 0 }%
      ifelse %
      0.5 exp %
      dup %
      t add exch %
      neg t add %
    }%
    \parametricplot{#2}{#3}{%
      t 2 exp %
      2.718281828 t 2 mul #1 sub exp %
      sub %
      dup 0 ge %
      {}%
      { pop 0 }%
      ifelse %
      0.5 exp neg %
      dup %
      t add exch %
      neg t add %
    }%
  }%
  \definecolor{mygray10}{gray}{0.9}%
  \definecolor{mygray20}{gray}{0.8}%
  \definecolor{mygray30}{gray}{0.7}%
  \definecolor{mygray50}{gray}{0.5}%
  \definecolor{mygray70}{gray}{0.3}%
  \definecolor{mygray80}{gray}{0.2}%
  \definecolor{mygray90}{gray}{0.1}%
  \psaxes[tickstyle=bottom]{->}(0,0)(10,10)%
  \uput{3mm}[270]{0}(10,0){$\upsilon$}%
  \uput{3mm}[180]{0}(0,10){$\zeta$}%
  \psset{plotpoints=1000}%
  \psset{linecolor=mygray10}%
  \parplot{2.001}{0.9687096746}{1.031956985}%
  \psset{linecolor=mygray20}%
  \parplot{3}{0.3017095627}{2.357676674}%
  \psset{linecolor=mygray30}%
  \parplot{4}{0.1585943396}{3.146193221}%
  \psset{linecolor=mygray50}%
  \parplot{5}{0.08979707022}{3.847396751}%
  \psset{linecolor=mygray70}%
  \parplot{6}{0.05246909746}{4.505241496}%
  \psset{linecolor=mygray80}%
  \parplot{7}{0.03115292702}{5.136340948}%
  \psset{linecolor=mygray90}%
  \parplot{8}{0.01866062909}{5.749031386}%
\end{pspicture}
%
%